\documentclass[12pt,a4paper]{article}
\usepackage[T1]{fontenc}

\usepackage{setspace}
\usepackage{amssymb}
\usepackage{indentfirst,mathrsfs}
\usepackage{amsfonts,amsmath,amsthm}
\usepackage{latexsym,amscd,multirow}
\usepackage{algorithm}
\usepackage{algpseudocode}
\usepackage{amsbsy}
\usepackage{xcolor}
\usepackage{diagbox}
\usepackage{colortbl}
\usepackage{multirow}
\usepackage{makecell}
\usepackage{float}
\usepackage{cases}
\usepackage{dsfont}
\usepackage{enumerate}
\usepackage{cite}
\usepackage{pifont}
\usepackage{graphicx}
\usepackage{mathtools}
\usepackage{tikz}
\usetikzlibrary{decorations,arrows}
\usetikzlibrary{decorations.markings}
\usetikzlibrary{decorations.pathreplacing}
\usetikzlibrary {datavisualization.formats.functions}

\usepackage{hyperref}
\hypersetup{hidelinks,
	backref=true,
	pagebackref=true,
	hyperindex=true,
	breaklinks=true,
	colorlinks=true,
	linkcolor=black!60,
	citecolor=black!50,
	urlcolor=blue,
	bookmarks=true,
	bookmarksopen=true
}
\usepackage{url}
\newtheorem{thm}{Theorem}
\newtheorem{lem}{Lemma}

\newtheorem{Open problem}{Open Problem}
\newtheorem{conjecture}{Conjecture}

\newtheorem{prop}{Proposition}

\newtheorem{defn}{Definition}
\newtheorem{example}{Example}
\newtheorem{remark}{Remark}

\newcommand{\nn}{\color{black}}

\newcommand{\y}{\textbf{y}}
\newcommand{\s}{\textbf{s}}

\allowdisplaybreaks[4]  %允许多行公式环境（如align、gather等）在页面中间断行

\begin{document}
		%\begin{CJK*}{GBK}{song}
	\title{
On the Construction and Correlation Properties of Permutation-Interleaved Zadoff-Chu Sequences
%Correlation Properties of Interleaved
%Zadoff-Chu Sequences from Permutation Polynomials
% Constructions of CAZAC sequences from interleaved Zadoff-Chu sequences by permutation polynomials and its aperiodic auto-correlation property
	}
	\author{
Qin Yuan
	\thanks{Q. Yuan
is with Faculty of
		Mathematics and Statistics, Hubei University, Wuhan, 430062, China, and
is with the School of Mathematics, Southwest Jiaotong University, Chengdu, 611756, China. Email:
	\href{mailto:yuanqin2020@aliyun.com}{yuanqin2020@aliyun.com}},		
Chunlei Li
		\thanks{C. Li is with the Department of Informatics, University of Bergen, Bergen, N-5020, Norway.
			Email: \href{mailto:chunlei.li@uib.no}{chunlei.li@uib.no}},
Xiangyong Zeng$^{\,*}$
\thanks{ X. Zeng is with
		Key Laboratory of Intelligent Sensing System and Security (Hubei University),
		Ministry of Education, Hubei Key Laboratory of Applied Mathematics, Faculty of
		Mathematics and Statistics, Hubei University, Wuhan, 430062, China. Email:
		\href{mailto:xiangyongzeng@aliyun.com}{xiangyongzeng@aliyun.com} },
	}
	
	\date{}
	\maketitle

\begin{abstract}
	Constant amplitude zero auto-correlation (CAZAC) sequences are
widely applied in waveforms for radar and communication systems.
Motivated by a recent work [Berggren and Popovi\'{c},
IEEE Trans. Inf. Theory 70(8), 6068-6075 (2024)], this paper
further investigates the approach to generating CAZAC sequences by
interleaving Zadoff-Chu (ZC) sequences with permutation polynomials (PPs).
We propose one class of high-degree PPs over the integer ring $\mathbb{Z}_N$, and utilize them and their inverses to interleave ZC sequences for
  constructing
 CAZAC sequences.
It is known that a CAZAC sequence can be extended to an equivalence class by five basic opertations.
We further show that the obtained CAZAC sequences are not covered by the equivalence classes of ZC sequences and interleaved ZC sequences  by quadratic PPs and their inverses,
and prove the sufficiency of the conjecture by Berggren and Popovi\'{c} in the aforementioned work.
In addition, we also evaluate the aperiodic auto-correlation of certain ZC sequences from quadratic PPs.

\end{abstract}

%%%%%%%%%%%%%%%%%%%%%%%%%%%%%%%%%%%%%%%%%%%%%%%%%%%%%%%

\section{Introduction}
Constant amplitude zero auto-correlation (CAZAC) sequences, which are unimodular and have zero out-of-phase periodic auto-correlation,
possess attractive properties and have been extensively studied \cite{Trends}.
A complex-valued sequence $\s=(s(0), s(1), \dots, s(N-1))$ is a CAZAC sequence if and only if
one of the following equivalent  statements holds:
%1. its normalized  DFT $\widehat{\s}$ is a CAZAC sequence;\\
%2. the sequence $\s$ is a bi-unimodular sequence (each $|\s(m)|=|\widehat{\s}(n)|=1$);\\
%3. the circulant matrix $H_s$ with first row $\s=(s(0), s(1), \dots, s(N-1))$ is a Hadamard matrix of order $N$.\\
%$$
%\begin{cases}
(i) its normalized  discrete Fourier transform (DFT) sequence $\widehat{\s}$ is a CAZAC sequence; or
(ii)  $\s$ is a bi-unimodular sequence, i.e., entries in both $\s$ and $\widehat{\s}$ have modulus one; or
%	\text{(ii)  the sequence $\s$ is a bi-unimodular sequence (each $|\s(m)|=|\widehat{\s}(n)|=1$);}\\
(iii) the circulant matrix $H$ generated from $\s$ is a complex Hadamard matrix, i.e., $HH^{\dagger} = NI_N$, where $H^{\dagger}$ is the conjugate transpose of $H$ and $I_N$ is the identity matrix of order $N$.
%\end{cases}
%$$
Moreover, there exists a one-to-one correspondence between unimodular cyclic $N$-roots $(z_0,z_1,\dots,z_{N-1})$ (which satisfies $\prod_{i=0}^{N-1}(x+z_i) = x^N+1$) and CAZAC sequences $\s$ of length $N$, which is given by
$(z_0,z_1,\dots,z_{N-1})=\bigg(\frac{s(1)}{s(0)},\frac{s(2)}{s(1)},\dots,
\frac{s(N-1)}{s(N-2)},\frac{s(0)}{s(N-1)} \bigg)
$ with the first term $s(0)=1$.
In engineering applications, the constant amplitude (CA) property facilitates transmitter peak power operation and provides theoretical immunity to amplitude distortions induced by additive noise; the zero auto-correlation (ZAC) property ensures minimal inter-signal interference in shared channels. These features make CAZAC sequences particularly effective in applications such as target recognition in radar and address synchronization in cellular access technologies~\cite{Trends, Ye_2022_, aperiodic_LM}.

Instances of CAZAC sequences \cite{Mow} include Zadoff-Chu (ZC) sequences \cite{Chu_1972_}, generalized Frank sequences \cite{Frank, Kumar_G_Frank}, generalized chirp-like (GCL) sequences \cite{Popovic_1992_}, Milewski sequences \cite{Milewski},
polyphase sequences constructed from generalized bent functions by Chung and Kumar \cite{Kumar}, modulating perfect sequences \cite{Popovic_2010_},
Björck sequences  \cite{Bjorck, Fourier_Bjorck},
biphase CAZAC sequences  \cite{Biphase_1990, Biphase_Golomb},
cubic polynomial phase sequences \cite{cubic_perfect} and    perfect Gaussian integer sequences based
 on cyclic difference sets \cite{Chen_Li}.
Mow proposed a unified construction of CAZAC sequences and conjectured that it covers all CAZAC sequences \cite{Mow_unif}.
As a class of commonly used CAZAC sequences,
ZC sequences are the primary manifestation of spread spectrum in modern cellular systems, including Long Term Evolution (LTE) and Fifth Generation New Radio (5G NR) \cite{5G_NR, LTE}.
%They have replaced the earlier pseudo-noise sequences and Walsh sequences, which were the mainstays of 3G cellular and IS-95.

Interleaving is a powerful technique in sequence design~\cite{Gong,Zhou_inter, Tang_ZCZ_interleave} and coding theory~\cite{Takeshita,Sun_2005_}.
The interleaving technique introduced by Gong \cite{Gong} has been widely adopted to construct low-correlation sequences from shorter constituent sequences, as demonstrated in works such as the works in \cite{Gong, Zhou_inter, Tang_ZCZ_interleave}.
In contrast, within the context of coding theory, permutation-based interleavers~\cite{Sun_2005_} were employed in turbo codes to achieve their remarkable error-correction performance~\cite{Takeshita}. Subsequently, several studies~\cite{Takeshita_2007_,Jonghoon2006,Rosnes_2012_,Berggren_2023_,Trifina_2016_} investigated turbo codes interleaved by quadratic permutation polynomials (QPPs), their inverses, and higher-degree permutation polynomials, as well as DFT-spread orthogonal frequency-division multiplexing (DFT-S-OFDM) interleaved by QPPs.
Fundamentally, both interleaving approaches rearrange an input sequence according to a prescribed permutation to produce an output sequence with desired properties. However, despite their success in turbo codes, permutation-based interleavers were not adopted in sequence design until the recent work of
 Berggren and Popovi\'{c} \cite{Berggren_2024_}.
 They showed that ZC sequences interleaved by QPPs and their inverses are also CAZAC sequences, and reported that the CAZAC property is not maintained in general when using cubic PPs.
 Based on numerical results, they conjectured that
 QPP-interleaved sequences are inequivalent to ZC sequences if and only if the sequence length
 $N$ is a prime power $p^n$ with prime $p\geq 5$ and $n \geq 2$.

%Generalised bent functions \cite{Kumar} and
%secondary constructions based on existing CAZAC sequences largely produce the currently known CAZAC sequences (see for example \cite{Trends} and references therein).
%This paper shall investigate the interleaving technique in the generation of CAZAC sequences.

% In coding theory,  a permutation-based interleaver \cite{Sun_2005_} was deployed in turbo codes to achieve their revolutionary performance \cite{Takeshita}.

%This mathematical technique of interleaving, shares the same fundamental principle of restructuring sequences with another way of interleaving by permutation polynomials,
% though they are applied differently.
% their applications diverge.

%The connection between the two methods lies in the permutation operation. The permutation-based interleaver can be regarded as an extension of the crucial permutation step in Gong's interleaving technique.

%\textbf{Motivations and Contributions}
 Motivated by the work of Berggren and Popovi\'{c} \cite{Berggren_2024_},  this paper will
 further investigate the approach to generating CAZAC sequences by interleaving ZC sequences
from PPs over integer rings $\mathbb{Z}_N$,
and also evaluate the aperiodic auto-correlation property of some of the obtained CAZAC sequences, a critical performance metric in practical applications \cite{Mow_ap, aperiodic_Schdmit, aperiodic_HV, aperiodic_Katz, Barker, Turyn, Fan_Frank,  Zhang_Golomb, Fan_Darnell, Gabidulin}.
Write the polyphase ZC sequence of length $N$ as $\s = (\xi_N^{f(k)})_{0\leq k<N}$ as given in~\eqref{ZC_seq} and the interleaved ZC sequence as $\s\circ\pi=(\xi_N^{f(\pi(k))})_{0\leq k<N}$ for a  permutation $\pi$ of $\mathbb{Z}_N$,
where $N$ is a positive integer.
In this paper, we first discuss the bijective property of a class of high-degree polynomials of the form  $\pi(x)=x^{p}+a x+b$, where $p$ is a prime divisor of $N$,
and then consider the periodic auto-correlation of $\s\circ\pi$ based on an interesting property of the second-order  finite difference  of the polynomial $f(\pi(k))$ on $\mathbb{Z}_N$ (as in the proof of Prop.~\ref{corre_thm}).
Furthermore, by imposing extra conditions on $\pi(x)$, we show that the interleaving sequences $\s\circ \pi$ and $\s\circ \pi^{-1}$ exhibit the ZAC property.
 The analysis for the sequence  $\s\circ \pi$ relies on the bijective property of the first-order  finite difference  of $f(\pi(k))$ on $\mathbb{Z}_{N_1}$, where $N_1=N/\gcd(d,N)$ for a nonzero shift $0<d<N$ (as in Prop.~\ref{prop_PP}),  and that for $\s\circ \pi^{-1}$ leverages the DFT and the bijective property of the first-order  finite difference  of $f(\pi(k))$ on $\mathbb{Z}_{N_1}$.
The techniques also enable us to prove the ZAC property of $\s\circ \pi$  and $\s\circ \pi^{-1}$ for arbitrary PPs $\pi(x)$ over $\mathbb{Z}_{2^n}$.
Besides, we also confirm the sufficiency of the conjecture by Berggren and Popovi\'{c} (in Prop.~\ref{prop1}) and
show that some of the CAZAC sequences $\s\circ \pi$  and $\s\circ \pi^{-1}$ derived in this paper
are not covered
by the equivalence classes of ZC sequences and interleaved ZC sequences  by QPPs  and QPPs' inverses (in Prop.~\ref{pro2}), with respect to the five basic operations.
The derived permutation-interleaved ZC sequences and their inequivalence relation are summarized in Table~\ref{tab_unique_seq}.
Finally, we analyze the aperiodic auto-correlation of some QPP-interleaved ZC sequences, showing that their magnitude is in the order of $O(N^{\frac{3}{4}})$ when $N=2^n$.
\begin{table}[!t]
%\small
%\scriptsize
\footnotesize
\caption{%New CAZAC sequences inequivalent to ZC sequences $\s$}
 Inequivalence among permutation-interleaved ZC sequences $\s$}
\label{tab_unique_seq}
	\begin{center}
\begin{tabular}{|c|c|c|c|c|}
\hline
    $N$      & PPs $\pi$ over $\mathbb{Z}_N$ &   CAZAC seq.s  & Inequivalence &  Ref.
\\[5pt] \hline \hline
    $N=p^n,$ &  $f_2x^2+f_1x+f_0,$               &   $\s \circ \pi$  & all $\s \circ \pi$  &  \cite{Berggren_2024_}
\\ [5pt]
    $p \geq 5, n \geq 2$ &  $p|f_2, p \nmid f_1$  &   $\s \circ \pi^{-1}$  & all $ \s \circ \pi^{-1}$  &   This paper
 \\ [6pt]
\hline
%%%%%%%%%%%%%%%%%
    $N=p^nq_1q_2 \dots q_r, \, p\geq 3,$  &   $x^p+ax+b, $    &   $\s \circ \pi$
&  all $\s \circ \pi$   &
 \\[5pt] $(q_i-1) \,|\, p-1, \forall  i, \, n\geq 2$  &   $ a \not\equiv 0,-1 \,(\mathrm{mod}\,\, p),$    &  $\s \circ \pi^{-1}$ &   all $ \s \circ \pi^{-1}$ &  This paper
\\    &   $ a \not\equiv -1 \,(\mathrm{mod}\,\, q_i), \forall  i$      &    &    &
\\[5pt]
 \hline
%%%%%%%%%%%%%%%%%
			$N=2^n,\, n\geq 4$ & $\sum\limits_{i=0}^{m}a_ix^i$  &  $\s \circ \pi$  &
 certain $\s \circ \pi$   & This paper
  \\[5pt]
			& &  certain $\s \circ \pi^{-1}$   & &
\\	\hline
\end{tabular}
	\small{
		\begin{itemize}
			\item  Distinct primes $p$ and  $q_i$, $i=1,2,\dots,r,$ where $r \geq 0$, and
when $r=0$, $N$ is deemed as $p^n$.

		\end{itemize}
	}
	\end{center}
\end{table}

%The derivation method is to use the Discrete Fourier Transform and to prove  exponent polynomials being PPs.

%Note that when $N=2^n$,  we investigate
%all PPs and their inverses over $\mathbb{Z}_{N}$ interleaving ZC sequences.
%Table \ref{tab_unique_seq} lists obtained inequivalent CAZAC sequences \cc  (refer to Definition \ref{def:inequivalent_cazac} ) \nn by PPs $\pi$ interleaved ZC sequences $\s$.
% When $N=p^n$, we prove that QPPs and their inverses interleaved ZC sequences in
%  \cite{Berggren_2024_} are inequivalent to ZC sequences, which gives a positive answer of
% the sufficiency of conjecture in \cite{Berggren_2024_}.
%When $N=p^n M$,
%PPs  $x^p+ax$
%interleaved ZC sequences we proposed  are further
%inequivalent to QPP-interleaved ZC sequences.
%%cannot be generated by compositions of the five operations in \eqref{mathe_opera} and QPPs/inverse QPPs on ZC sequences.
%Furthermore, we determine an upper bound $ 1.686012 \sqrt{N}$ of the aperiodic auto-correlation of  some interleaved CAZAC  sequences without restrictions on root indices.

The remainder of this paper is organized as follows.
In Section \ref{Sec2}, we introduces some necessary notations, definitions and auxiliary lemmas.
Section  \ref{Sec3} studies PPs of the form  $\pi(x)=x^{p}+a x+b$  over $\mathbb{Z}_N$  for any positive integer $N$ and a prime factor $p$ of $N$, and the periodic auto-correlation of interleaved ZC  sequences $\s\circ \pi$.
In Section \ref{Sec_New_4}, we propose CAZAC sequences by interleaving ZC sequences with PPs and their inverses.
  Section \ref{Sec4} shows that some proposed CAZAC sequences are not covered by the equivalence  classes of ZC sequences and interleaved ZC sequences  by QPPs  and their inverses.
  % QPP-interleaved ZC sequences.
  In Section \ref{Sec6666}, we derive an upper bound of the aperiodic auto-correlation of certain QPP-interleaved ZC sequences.
Finally, Section \ref{Sec5} concludes
the work of this paper.

\section{Preliminaries}\label{Sec2}
We first list some notations which will be used throughout the paper:
\begin{itemize} \setlength{\itemsep}{0.4pt}

%\item  $M = q_1 q_2 \cdots q_r$ with $r\geq 0$ and different primes
%$q_1 <q_2 <\dots <q_r$,  and  when $r=0$, $M$ is deemed as $1$;

%\item  $N$ denotes a positive integer;

\item  $\mathbb{Z}_{N}=\{0,1,2,\dots,N-1\}$ denotes the integer ring modulo a positive integer $N$;
	
\item	$\xi_N$   denotes the primitive $N$-th root of unity $e^{ \frac{ -2\pi\sqrt{-1} }{N}}$;

\item $|a|$ and $a^*$ denote the magnitude and conjugate of a complex
	number $a$, respectively;
	
%\item  $\gcd(a,b)$ denotes the greatest common divisor of $a$ and $b$;
	
\item
$\s$
denotes a sequence of length $N$,
whose $k$-th term is denoted as $s(k)$ for  $0\leq k<N$;
%and $s(k)$ denotes its $k$-th item with $0\leq k<N$;

\item $R[x]$ denotes the polynomial ring over $R$;	

\item $\pi $  denotes a permutation polynomial (PP) in $\mathbb{Z}_{N}[x]$;

\item
$  \s\circ \pi$
denotes the interleaved sequence  by the permutation $\pi$ of  $\mathbb{Z}_{N}$, namely,
the $k$-th term of   $\s\circ \pi$ is given by $({s}\circ \pi) (k)= s(\pi(k))$ for $0\leq k<N$;
%and its $k$-th term is denoted as $({s}\circ \pi) (k)= s(\pi(k))$ for $0\leq k<N$;

	\item
	$\widehat{\s}$
	denotes the normalized DFT sequence of  $\s$;

	\item $\binom{n}{m}$ is a binomial coefficient given by $\binom{n}{m}=\frac{n!}{m!(n-m)!}$ with $0\leq m\leq n$;

\item   The periodic auto-correlation of a sequence $\s$ of length $N$ at a shift $d$ is given by  $$\theta(d) =\sum^{N-1}_{k=0}s(k)s^{*}(k + d), 0\leq d<N,$$
where $k + d$ is taken modulo $N$;
%(mod $N$) denotes the modulo $N$ operator;

\item   The aperiodic auto-correlation of a sequence $\s$ of length $N$ at a shift $d$  is given by  $$\tilde{\theta}(d) =\sum\limits_{k=0}^{N-d-1} s(k)s^{*}(k+d),   0\leq d<N. $$

\end{itemize}

\subsection{Permutation polynomials over $\mathbb{Z}_{N}$}\label{SEC_2.1}
For a polynomial $f(x) = a_mx^m + \cdots + a_1x + a_0$ in $\mathbb{Z}_N[x]$, its (formal) derivative is given by $f'(x) =ma_mx^{m-1}+\dots +3a_3x^2 + 2a_2x  + a_1 \in \mathbb{Z}_N[x]$.  The first-order finite difference of $f(x)$ is given by $\Delta_{d}=f(x+d)-f(x)$, and the second-order finite difference of $f(x)$ is given by  $\Delta_{d_1,d_2}f(x)=\Delta_{d_2}[\Delta_{d_1}   f(x)] = f(x+d_1+d_2) - f(x+d_1) - f(x+d_2) + f(x).$
%\begin{equation}\label{second_difference}
%\Delta_{d_1,d_2}f(x)=\Delta_{d_2}[\Delta_{d_1}   f(x)] = f(x+d_1+d_2) - f(x+d_1) - f(x+d_2) + f(x).
%\end{equation}

Given  nonnegative integers  $m$ and $N\geq 2$, a polynomial
%\begin{equation*}%\label{PPP}
$\pi(x)= a_mx^m + \dots + a_2x^2  +  a_1x + a_0$
%\end{equation*}
with  $a_0,a_1,\dots,a_m \in \mathbb{Z}_{N}$ and $a_m \neq 0$ is said to be an $m$th-degree permutation polynomial (PP) in  $\mathbb{Z}_{N}[x]$, when
$\{ \pi(k)\,\,(\mathrm{mod}\,\, N): 0\leq k < N\}$ is
%$\pi\, (\mathrm{mod}\, N)$ induces
a permutation of $\mathbb{Z}_{N} $.	
The compositional inverse of $\pi$, denoted by   $\pi^{-1}$,
is the PP such that   $\pi^{-1}\circ \pi(k)=\pi^{-1}( \pi(k)) =k$ for any $k$.
For a PP over $\mathbb{Z}_N$,
the polynomial representation of its inverse always exists, but is not necessarily unique due to the existence of zero  divisors  in $\mathbb{Z}_N$ for a composite integer $N$ \cite{Lidl_Niederreiter_1996, Jonghoon2006, Lahtonen2012}.
Linear PPs over $\mathbb{Z}_N$ are trivially of the form $a_1x+a_0$ with $\gcd(a_1,N)=1$.
Quadratic PPs (QPPs) $\pi(x)=a_2x^2+a_1x+a_0$ over $\mathbb{Z}_N$ have their coefficients $a_2$ and $a_1$ satisfying  the following conditions:
\begin{equation*}%\label{coeff}
\left\{
  \begin{array}{ll}
    \text{ if }  v_2(N) = 1, & \text{ then }  a_2+a_1 \text{ is odd},\gcd(a_1,N/2)=1,  \\
&  \text{ and each prime divisor of } N/2 \text{ divides } a_2;\\
    \text{ if }  v_2(N) \neq 1, & \text{ then } \gcd(a_1,N)=1, \text{ each prime divisor of } N \text{ divides } a_2,
  \end{array}
\right.
\end{equation*}
where $v_2(N)$ is the 2-adic value of $N$, namely, the largest power of 2 that divides $N$.

%不确定放这个位置是否合适
%\cc
%Note that for any positive integer $N$, it can be
%represented as $N=q_1^{n_1}\dots q_r^{n_r}$, where $q_1<q_2<\dots <q_r$ are distinct primes and $n_i$'s are positive integers with $i=1,2,\dots,r$.
%\nn
For a higher degree, we recall necessary and sufficient  conditions for a polynomial $\pi(x)$ over $\mathbb{Z}_{N}$
to be a PP below.

	\begin{lem}	(\!\!\cite{Rivest}) \label{N=2^n}
		Let $N=2^n$ with  $n\geq 2$ and
		$\pi(x)=a_mx^m + a_{m-1}x^{m-1} + \dots + a_2x^2 + a_1x + a_0\in \mathbb{Z}_{N}[x]$.
		Then the  polynomial  $\pi$ is a PP over $\mathbb{Z}_{2^n}$ if and only if
	$$a_1 \equiv 1 \,\,(\mathrm{mod}\,\, 2), \, \, a_2 + a_4 + a_6 + \dots \equiv 0 \,\,(\mathrm{mod}\,\, 2) \text{ and\, }  a_3 + a_5 + a_7 + \dots \equiv 0 \,\,(\mathrm{mod}\,\, 2). $$
	\end{lem}

%For $N=2^n$, there exist cubic and quartic permutation polynomials that do not reduce to quadratic  or linear permutation polynomials \cite{Trifina_2016_}.

\begin{lem}\label{N=p^n}
	\noindent{\rm (i)}	(\!\!\cite{Har_1975})
For a prime $p$ and $n\geq 2$,	a polynomial $\pi$ is a PP over $\mathbb{Z}_{p^n}$  if and only if $\pi$ is a PP
	over $\mathbb{Z}_{p}$ and $\pi(k)'\neq 0 \,(\mathrm{mod}\,\, p)$ for all integers $k\in
\mathbb{Z}_{p}$.
%\mathbb{Z}_{p^n}$.

	\noindent{\rm (ii)} (\!\!\cite{Sun_2005_})
	For $N=\prod\limits_{i=1}^r q_i^{n_i}$ with integers $n_i >0$ and distinct primes $q_i$'s,
	a polynomial $\pi$ is a PP over $\mathbb{Z}_N$ if and only if $\pi$ is a PP
	over $\mathbb{Z}_{q_i^{n_i}}$ for any $ i$.
\end{lem}

%\subsection{Periodic auto-correlation}

\subsection{CAZAC sequences and their equivalence}	
%A sequence $\s=(s(0),s(1),\dots,s(N-1))$ has constant amplitude, i.e., $|s(k)|= 1, 0 \leq  k<N$, and zero out-of-phase periodic auto-correlation function
%is known as a constant amplitude zero auto-correlation (CAZAC) sequence.
\begin{defn}[CAZAC sequences]
A sequence $\s=(s(0),s(1),\dots,s(N-1))$ of length $N$ is a constant amplitude zero auto-correlation (CAZAC) sequence if it satisfies the following:
\begin{description}%[itemsep=0.4pt]
  \item[(CA)]  $|s(k)|=1$ for any $0\leq k<N$; %$k \in \mathbb{Z}_{N}^*$;
  \item[(ZAC)]$\theta(d) =\sum\limits^{N-1}_{k=0}s(k)s^{*}(k + d)=0$ for any $ 0< d<N$,
where $k + d$ is taken modulo $N$.
\end{description}
%\begin{itemize}
%  \item (CA) $|s(k)|=1$ for any $0\leq k<N$; %$k \in \mathbb{Z}_{N}^*$;
%  \item (ZAC) $\theta(d) =\sum\limits^{N-1}_{k=0}s(k)s^{*}(k + d)=0$ for any $ 0\leq d<N$,
%where $k + d$ is taken modulo $N$.
%\end{itemize}
\end{defn}
A CAZAC sequence is also referred to as a perfect root-of-unity sequence~\cite{Mow_unif}.
As pointed out in \cite{Benedetto_2007_, Trends}, a CAZAC sequence $\s$ of length $N$ under the following mathematical operations remains a CAZAC sequence:
\begin{equation}\label{mathe_opera}
\left\{
\begin{aligned}
&\text{(Rotation) } {y}(k)=cs(k), \text{ for a complex number } c \text{ with }  |c|=1;  \\
&\text{(Translation) } {y}(k)={x}(k+d \,( \mathrm{mod}\, N))  \text{ for an integer } d;  \\
&\text{(Decimation) } {y}(k)={x}(ak \,( \mathrm{mod}\, N)),  \text{ where } a \text{ is an integer with } \mathrm{gcd}(a,N)=1;  \\
&\text{(Linear frequency modulation) } {y}(k)=\xi_N ^{lk} s(k) \text{ for an integer } l;  \\
&\text{(Conjugate) } {y}(k)= {x}^*(k),
\end{aligned}
\right.
\end{equation}
where $k=0,1,\dots,N-1$.  Based on these operations, we introduce the equivalence class of a CAZAC sequence. %inequivalent to ZC sequences.

\begin{defn}[Equivalence Classes]\label{def:equivalence_class}
	An \textit{equivalence class of a CAZAC sequence} is the set of all sequences generated by composing the five operations in \eqref{mathe_opera} on a given CAZAC sequence: rotation, translation, decimation, linear frequency modulation, and conjugation.
\end{defn}

%where $g(k)=\frac{1}{2} [\pi^2(k)+(N\mathrm{mod}\,\,2 )\pi(k) ]$.

%For $N = sm^2$ with integers $s$ and $m$,
%a generalized chirp-like (GCL) sequence \cite{Popovic_1992_} is defined as
%\begin{equation*}%\label{GCL_seq_exp}
%\s \cdot \textbf{w} =({x}(0){w}(0 ),{x}(1){w}(1 ) ,\dots,{x}(N-1){w}(N-1 \,(\mathrm{mod}\,\,  m))  ),
%\end{equation*}
%where $\s$ is a ZC sequence and $\textbf{w}$  is an arbitrary unimodular complex ``modulation'' sequence with period $m$. Note that GCL sequences as modulated
%ZC sequences
%are also CAZAC sequences.

In general, achieving zero periodic correlation across the entire delay range is
difficult when the phase of sequences is a restrictive consideration in practical applications.  To address this limitation, zero-correlation zone (ZCZ) sequences   were introduced \cite{Tang_Fan_Matsufuji} for quasi-synchronous systems. These sequences have zero out-of-phase periodic auto-correlation within a bounded delay range $|d| \leq D$, where $D$ denotes the ZCZ length.

\subsection{Permutation-interleaved Zadoff-Chu sequences}

We first give the definition of permutation-interleaved sequences.
%interleavers by using PPs.

\begin{defn} (\!\!\cite{Sun_2005_}) \label{PPinterleaved}
	Let $\s =(s(0),s(1),\dots, s(N-1)) $ be a sequence of length $N$,
	and let  $\pi$ be a permutation of $\mathbb{Z}_{N}$.
	Then the permutation-interleaved sequence, denoted by $  \s\circ \pi$,
	is defined by
	$$({s}\circ \pi) (k)= s(\pi(k)) \text{ for } 0\leq k<N.$$
\end{defn}

\begin{remark}
A permutation-based interleaver is different from
the interleaving technique introduced by Gong \cite{Gong}. The latter's  core idea is to interleave multiple short-period base sequences (which can be shift-equivalent sequences), typically by arranging them in columns and reading the new sequence row by row, so as to obtain a longer-period sequence with good correlation properties.
%Moreover, these multiple short-period base sequences can be shift equivalent sequences.
\end{remark}

Many classes of CAZAC sequences were proposed in the literature \cite{Mow}, and the Zadoff-Chu sequence (ZC sequence) \cite{Chu_1972_} is probably one of the most widely deployed CAZAC sequences in radar and communication systems.
A ZC sequence is a polyphase sequence $\s=(s(0),s(1),\dots,s(N-1))$  defined by $s(k)=\xi_N^{f(k)},  k=0,1,\dots,N-1$ with
\begin{equation}\label{ZC_seq}
f(k) :={u(k^2+(N\,\mathrm{mod}\,\, 2)k+2lk)/2},
\end{equation}
where %$\xi_N = e^{ \frac{ -2\pi\sqrt{-1} }{N}}$ denotes a  primitive $N$-th root of unity,
$u$ is a root index satisfying $\mathrm{gcd}(u,N)=1$, and $l$ is an arbitrary integer\cite{Chu_1972_}.
The special case of the root index  $u=1$ yields the original ZC sequence.

For interleaved ZC sequences $\y=\s \circ \pi$ with any PP $\pi$ in $\mathbb{Z}_{N}[x]$, we have
%f\circ \pi(k)
\begin{equation*}%\label{yyyyk}
y(k)=s(\pi(k))=\xi_N^{f( \pi(k))}=\xi_N^{u(\pi(k)^2+\pi(k)(N \,\mathrm{mod}\, 2)+2l\pi(k))/2},
\end{equation*}
for $k=0,1,...,N-1$. It is clear that its magnitude satisfies $|y(k)| =1$, i.e., the CA property.
Hence we only need to consider the ZAC property of interleaved ZC sequences.
Recall that the periodic auto-correlation function of $\textbf{y}=\s \circ \pi = (\xi_N^{f( \pi(k))})_{0\leq k<N}$ is given by
$$\theta(d) =\sum^{N-1}_{k=0}y(k)y^{*}(k + d )=\sum^{N-1}_{k=0}\xi_N^{f( \pi(k))- f( \pi(k+d))}.$$
%PPs-interleaved  ZC sequences $\textbf{y}=\s \circ \pi$ satisfy
%$$y(k)=s(\pi(k))=\xi_N^{\frac{u}{2} (\pi^2(k)+(N\,\mathrm{mod}\,\,2 )\pi(k)+2l\pi(k) ) }, k=0,1,\dots,N-1.$$
Here  $\xi_N^{f( \pi(k))}$ can be  rewritten as  $	\xi_N^{f( \pi(k))}=\xi_N^{-\frac{u}{2} (l+(N\,\mathrm{mod} \,2))l } \xi_N^{\frac{u}{2} (\pi(k)+l+(N\,\mathrm{mod} \,2))(\pi(k)+l)}.
$
This indicates that the parameter $l$ introduces a rotation and a translation, which are mathematical operations
in \eqref{mathe_opera} that preserve the CAZAC property. Without loss of generality, we will set $l=0$ in our subsequent investigation.
Let $\zeta_N=\xi_N^{-u}=\xi_N^{N-u}$. Then $\zeta_N$ is a primitive $N$-th root of unity since $\mathrm{gcd}(u,N)=1$.
Hence for $0<d<N$, we rewrite the periodic auto-correlation of $\textbf{y}=\s \circ \pi$ as
%\begin{align}\label{core_d}
%\theta(d) %&=\, \sum^{N-1}\limits_{k=0}y(k)y^{*}(k + d)
%&=\sum\limits^{N-1}_{k=0}\xi_N^{f( \pi(k))- f( \pi(k+d))}  \nonumber \\
%&=\, \sum^{N-1}\limits_{k=0} \xi_N^{\frac{u}{2}(\pi^2(k) +(N\mathrm{mod}\,\,2 )\pi(k)) -\frac{u}{2}(\pi^2(k+d)+(N\mathrm{mod}\,\,2 )\pi(k+d)) }  \nonumber \\
%&=\, \sum^{N-1}\limits_{k=0} \zeta_N^{\frac{1}{2} [\pi^2(k+d)-\pi^2(k)+(N\mathrm{mod}\,\,2 )(\pi(k+d)-\pi(k) ) ] }.
%\end{align}
\begin{equation}\label{core_d}
\begin{array}{cll}
\theta(d) %&=\, \sum^{N-1}\limits_{k=0}y(k)y^{*}(k + d)
&=\sum\limits^{N-1}_{k=0}\xi_N^{f( \pi(k))- f( \pi(k+d))}\\
&=\, \sum^{N-1}\limits_{k=0} \xi_N^{\frac{u}{2}(\pi^2(k)+(N\mathrm{mod}\,\,2 )\pi(k)) -\frac{u}{2}(\pi^2(k+d)+(N\mathrm{mod}\,\,2 )\pi(k+d)) }\\
&=\, \sum^{N-1}\limits_{k=0} \zeta_N^{\frac{1}{2} [\pi^2(k+d)-\pi^2(k)+(N\mathrm{mod}\,\,2 )(\pi(k+d)-\pi(k) ) ] }.
\end{array}
\end{equation}

The investigation of the above auto-correlation function  can be interpreted from the perspective of cyclotomic polynomials.
Recall that the $s$-th cyclotomic polynomial $\Phi_s(x)$ is the unique monic irreducible polynomial
over the field of rational numbers $\mathbb{Q}$.
The roots of $\Phi_s(x)$  are the primitive $s$-th roots of unity.
Precisely, $\Phi_1(x) = x-1$, $\Phi_p(x) = x^{p-1}+x^{p-2} + \cdots + x + 1$ for a prime $p$ and
for the general cases, we have
\begin{equation}\label{Eq-Cyclo}
\Phi_s(x) = \prod_{1\leq k<s, \gcd(k,s)=1}(x-\zeta_s^k) \text{ and } x^{N-1}+\cdots + x + 1 = \prod_{s|N, s\neq 1}\Phi_s(x),
\end{equation}
where $s$ and $N$ are positive integers and $\zeta_s=e^{2 \pi \sqrt{-1}/s}$ is a primitive $s$-th  root of unity.
Then for a polynomial $A(x) = \sum_{k=0}^{N-1} x^{a_k}$,  it follows that
$$
A(\zeta_N) = \sum_{k=0}^{N-1} \zeta_N^{a_k} = 0 \text{ if and only if  } \Phi_N(x) | A(x).
$$
According to the work of Lam and Leung~\cite{Lam-Leung}, % https://doi.org/10.1006/jabr.1999.8089
which discussed vanishing sums of $N$-th roots of unity,
there are many choices of sequences $(a_0,a_1,\dots, a_{N-1})$ such that $A(\zeta_N) = 0$.
We have one obvious type of such sequences: $(a_0,a_1,\dots, a_{N-1})$ after being permuted is composed of multiple copies of $(0,1,2,\dots, N_1-1)$,
which corresponds to $A(x) = \frac{N}{N_1} (x^{N_1-1}+\cdots + x + 1)$, for a divisor $N_1$ of $N$.
The proofs of results in subsequent sections will be based on this type of instances.

%\subsection{New CAZAC sequences inequivalent to ZC sequences}

%As pointed out in \cite{Benedetto_2007_, Trends}, a CAZAC sequence $\s$ of length $N$ under the following mathematical operations remains a CAZAC sequence:
%\begin{equation}\label{mathe_opera}
%\left\{
%\begin{aligned}
%&\text{(Rotation) } {y}(k)=cs(k), \text{ for a complex number } c \text{ with }  |c|=1;  \\
%&\text{(Translation) } {y}(k)={x}(k+d \,( \mathrm{mod}\, N))  \text{ for an integer } d;  \\
%&\text{(Decimation) } {y}(k)={x}(ak \,( \mathrm{mod}\, N)),  \text{ where } a \text{ is an integer with } \mathrm{gcd}(a,N)=1;  \\
%&\text{(Linear frequency modulation) } {y}(k)=\xi_N ^{lk} s(k) \text{ for an integer } l;  \\
%&\text{(Conjugate) } {y}(k)= {x}^*(k),
%\end{aligned}
%\right.
%\end{equation}
%where $k=0,1,\dots,N-1$.
%Note that ZC sequences under the five operations in \eqref{mathe_opera} remain the CAZAC  property.

%A CAZAC sequence is deemed \textit{inequivalent to ZC sequences} if it does not belong to the equivalence class of ZC sequence.
Following the notion of  the  equivalence class, we are interested in finding new CAZAC sequences that are inequivalent to ZC sequences.
	For this purpose, we will show that for a rotation $r\in\{0,1,\dots,2N-1\}$, a translation $d\in \{0,1,\dots,N-1\}$, root indexes $u_i$  with $\mathrm{gcd}(u_i,N)=1$ and $i=1,2$,  a linear frequency  modulation sequence $\xi_N^{vk}, v=0,1,\dots,N-1$, and a complex conjugate operation $s\in\{-1,1\}$,
there exists a certain $k\in\{0,1,\dots,N-1\}$ such that
\begin{equation}\label{uuuuu}
	\xi_N^{u_1(\pi(k)+(N\,\mathrm{mod}\,2))\pi(k)/2} \neq \xi_N^{vk}\xi_N^{su_2(k+d+(N\,\mathrm{mod}\,2))(k+d)/2+r/2}.
\end{equation}
Before ending this subsection, we recall the conjecture by Berggren and Popovi$\mathrm{\acute{c}}$ \cite{Berggren_2023_}
 about inequivalent  QPP-interleaved sequences.
In Section~\ref{Sec4},  we shall prove its sufficiency and show that the inverses of QPPs
interleaved sequences are also  inequivalent to ZC sequences.
%we shall give a positive answer for  the sufficiency of the conjecture.

\begin{conjecture}(\!\!\cite{Berggren_2024_})\label{conj_Ber}
Given a ZC sequence $\s$,  all QPP-interleaved sequences $\s \circ \pi$ with  $\pi(x)=f_2x^2+f_1x +f_0$   are
 inequivalent to ZC sequences if and only if $N = p^n$ where $p\,(p \geq 5)$ is a prime and $n\,(n \geq 2)$ is an integer.
\end{conjecture}

%According to the above analysis, starting from a ZC sequence, the five basic operations together with interleaving by QPPs or the inverses of QPPs yield  a large  family  of CAZAC sequences. A natural question arises: can more  CAZAC sequences be obtained by applying  other PPs to ZC sequences?

\subsection{Auxiliary lemmas}
This subsection presents some lemmas which help determine the periodic auto-correlation of interleaved sequences.

We first recall several useful facts about binomial coefficients.

\begin{lem}(\!\!\cite{Rosen, Lucas})\label{parity}	
	%\leavevmode
	\noindent{\rm (i)} Let \( p \) be a prime. Then the following statements hold:\\
	$
		\begin{array}{ll}
			- & p \mid   \binom{p}{i} \quad \text{for } 1 \leq i < p; \\
			- & p \mid \binom{p+1}{i}  \quad \text{for } 2 \leq i < p; \\
			- & \text{If } p \geq 3,\text{ then } p \mid  \binom{2p}{j} \text{ for } 1 \leq j < 2p,\, j \neq p.
		\end{array}
	$

	\noindent{\rm (ii)} For integers \( n \) and $i$ with \( 0 \leq i \leq 2n \),
	one has
	$\binom{2n}{i}\equiv 0 \pmod{2} $ if $i$ is odd, and $\binom{2n}{i} \equiv \binom{n}{i/2} \pmod{2}$ if $i$ is even.

	\noindent{\rm (iii) [Lucas's Theorem]}
For a prime \( p \) and non-negative integers $ m, n $ with $m \geq n $, the following congruence relation holds:
\[
\binom{m}{n} \equiv \prod_{i=0}^{k} \binom{m_i}{n_i} \pmod{p},
\]
where
$
m = m_k p^k + m_{k-1} p^{k-1} + \cdots + m_1 p + m_0,
$
and
$
n = n_k p^k + n_{k-1} p^{k-1} + \cdots + n_1 p + n_0
$
are the base \( p \) expansions of \( m \) and \( n \), respectively.
\end{lem}

From Lemma \ref{parity}, the following congruence equality can be derived, which will be used to calculate  the auto-correlation function.
The proof of Lemma \ref{N=pq} is given in
\hyperref[appendix:A]{Appendix A}.

\begin{lem}\label{N=pq}
	Let $p$ be a prime, and let 	$m$ and $n$ be positive integers.
	Then we have
$$\frac{1}{2}[(m+n)^{2p}+(n\, \mathrm{mod}\,\,2 )(m+n)^{p}] \equiv \frac{1}{2}[m^{2p}+(n \, \mathrm{mod}\,\,2 )m^{p}] \ \;(\mathrm{mod}\,\,n ).$$
\end{lem}

\section{Permutation-interleaved ZC sequences}\label{Sec3}
In this section,  we
investigate the auto-correlation of interleaved ZC sequences $\s \circ \pi$
as long as there exists a PP $\pi(x) = x^{p}+a x+b$ in $\mathbb{Z}_{N}[x]$ when
$p$ is a prime factor of $N$. Moreover, we propose an explicit class of PPs in this form under an extra condition.

%Recall that the auto-correlation function of $\textbf{y}=\s \circ \pi$ is
%$$\theta(d) =\sum^{N-1}_{k=0}y(k)y^{*}(k + d ).$$
%%PPs-interleaved  ZC sequences $\textbf{y}=\s \circ \pi$ satisfy
%%$$y(k)=s(\pi(k))=\xi_N^{\frac{u}{2} (\pi^2(k)+(N\,\mathrm{mod}\,\,2 )\pi(k)+2l\pi(k) ) }, k=0,1,\dots,N-1.$$
%We first rewrite $	y(k)$ as $	y(k)=\xi_N^{-\frac{u}{2} (l+(N\,\mathrm{mod} \,2))l }\times \xi_N^{\frac{u}{2} (\pi(k)+l+(N\,\mathrm{mod} \,2))(\pi(k)+l)}.
%$
%It implies that $l$ introduces a rotation and a translation, which are mathematical operations
% in \eqref{mathe_opera} that preserve the CAZAC property. Without loss of generality, we can set $l=0$ in the proof of correlation properties.
%Let $\zeta_N=\xi_N^{-u}=\xi_N^{N-u}$. Then $\zeta_N$ is a primitive $N$-th root of unity since $\mathrm{gcd}(u,N)=1$.
%Hence for $0<d<N$, we have the auto-correlation function of $\textbf{y}=(y(0),y(1),\dots,y(N-1))$ as
%\begin{equation}\label{core_d}
%	\begin{array}{cll}
%	\theta(d) &=\, \sum^{N-1}\limits_{k=0}y(k)y^{*}(k + d)\\
%	&=\, \sum^{N-1}\limits_{k=0} \xi_N^{\frac{u}{2}(\pi^2(k)+(N\mathrm{mod}\,\,2 )\pi(k))} \xi_N^{-\frac{u}{2}(\pi^2(k+d)+(N\mathrm{mod}\,\,2 )\pi(k+d)) }\\
%	%&=\, \sum^{N-1}\limits_{k=0} \xi_N^{ u( \pi^2(k)/2+\pi(k)/2+q\pi(k)-\pi^2(k+d)/2-\pi(k+d)/2-q\pi(k+d) )  }\\
%	&=\, \sum^{N-1}\limits_{k=0} \zeta_N^{\frac{1}{2} [\pi^2(k+d)-\pi^2(k)+(N\mathrm{mod}\,\,2 )(\pi(k+d)-\pi(k) ) ] }.\\
%%&	=C_0\sum^{N-1}\limits_{k=0} \zeta_N^{ h_d(k)}.\\
%\end{array}
%\end{equation}

%With the above preparations, we present a result of this paper below.

\begin{prop}\label{corre_thm}	
Assume $N$ is factorized as $N=q_1^{n_1}q_2^{n_2}\dots q_r^{n_r}$ with primes $q_1<\cdots <q_r$.
 Let  $\s$ be a ZC sequence of length $N$. When the polynomial $\pi(x) = x^{q_i} + ax+b \in \mathbb{Z}_N[x]$ permutes $\mathbb{Z}_N$ for a certain integer $1\leq i\leq r$,
 the periodic auto-correlation of the interleaved  sequence $\s \circ \pi$ satisfies  $\theta(d)=0$ when $q_i^{n_i} \nmid d$.
\end{prop}

\begin{remark} Notice that the interleaved sequence in Proposition \ref{corre_thm} does not always exhibit the ZAC property.
When $q_i^{n_i} \, |\, d$,
experiments show that
the periodic auto-correlation of the interleaved sequence varies with different  choices  of $N$ and PPs, which indicates that extra constraints are needed for the ZAC property.
\end{remark}

%\subsection{Proof of Proposition \ref{corre_thm}}\label{Sec3-2}
%\subsection{Periodic auto-correlation values of PPs interleaving ZC sequences}\label{Sec3-2}
Take $\pi(x)=x^p+ax+b$.
Since $b$ is a constant translation, which is a mathematical operation in \eqref{mathe_opera} that preserves the CAZAC property, we can set $b=0$ in the proof of $\theta(d)$.
According to \eqref{core_d}, the periodic auto-correlation of the interleaved sequence $\textbf{y}=\s \circ \pi$ is given by
\begin{equation}\label{corehk}
	\theta(d) =\sum\limits^{N-1}_{k=0}\xi_N^{f( \pi(k))- f( \pi(k+d))}=C_0\sum^{N-1}\limits_{k=0} \zeta_N^{ h_d(k)},
\end{equation}
where $C_0=\zeta_N^{ ad(ad+1)/2} $ and
\begin{equation}\label{h_d(k_form)}
	h_d(k)=\frac{1}{2}[(k+d)^{2p}-k^{2p}  +(N\mathrm{mod}\,\,2 )((k+d)^{p}-k^{p}) ] +a[(k+d)^{p+1}-k^{p+1}] +a^2dk.
\end{equation}

The complete proof of Proposition~\ref{corre_thm} involves complicated calculations, so we provide only the main idea here and place the  details in \hyperref[appendix:B]{Appendix B}.

\noindent\textbf{Proof sketch of Proposition \ref{corre_thm}.}
Let $p$ be any prime factor of $N$ and $v_p(N)=n$.
According to \eqref{corehk}, we shall prove that for $p^n \nmid d$, the sum
$$\sum^{N-1}\limits_{k=0} \zeta_N^{ h_d(k)}=0,$$
	where $h_d(k)$ is given in \eqref{h_d(k_form)}.
This statement follows from the following two key observations:
\begin{enumerate}[(1)] \setlength{\itemsep}{0.5pt}
\item $h_d(k + N) \equiv h_d(k)  \ (\mathrm{mod}\,\,N)$ for any  $k \in \mathbb{Z}_N$, which implies $\sum_{k=0}^{N-1} \zeta_N^{h_d(k)} = \sum_{k=0}^{N-1}\zeta_N^{h_d(k+t)}$ for any positive integer  $t$;
\item  when $d$ is not divisible by $p^n$, we have an integer $t = \frac{2N}{\gcd(d,N)p} $ \text{ such that } the value $(h_d(k + t) - h_d(k)) \ \mathrm{mod}\,N$ for any  $k \in \mathbb{Z}_N$ is a nonzero constant $c$,  \text{ which implies} $\sum_{k=0}^{N-1} \zeta_N^{h_d(k+t)}  = \zeta_N^c \sum_{k=0}^{N-1} \zeta_N^{h_d(k)}$ for $c\neq 0$.
\end{enumerate}
%\begin{flalign}
%\text{1) } & h_d(k + N) \equiv h_d(k)  \ \mathrm{mod}\,\,N \text{ for any } k \in \mathbb{Z}_N,
% \text{which implies }
%\text{1) } & \qquad  \qquad  \qquad  \qquad h_d(k + N) \equiv h_d(k) \pmod{N} \text{ for any } k \in \mathbb{Z}_N, & \notag \\
%& \text{which implies that } \sum_{k=0}^{N-1} \zeta_N^{h_d(k)} = \sum_{k=0}^{N-1}
%\zeta_N^{h_d(k+t)} \text{ for any positive integer } t; & \notag \\
%\text{2) } & \text{for } p^n \nmid d, \text{ we have an integer } t = \frac{2N}{\gcd(d,N)p} \text{ such that } (h_d(k + t) - h_d(k)) \ \mathrm{mod}\,\,N  & \notag \\
%& \text{becomes a nonzero constant for any } k \in \mathbb{Z}_N, \text{ which implies} \sum_{k=0}^{N-1} \zeta_N^{h_d(k+t)}  = C_1 \sum_{k=0}^{N-1} \zeta_N^{h_d(k)}& \notag \\
%&  \text{where } C_1 \neq 1. & \notag
%\end{flalign}
Combining these observations leads to the desired statement.
$\text{\hfill} \square$

\begin{remark}\label{unique_ZCZ}
We can see that interleaved sequences $\s \circ \pi$ in
Proposition \ref{corre_thm} are ZCZ sequences with ZCZ length ${q_i}^{n_i}-1$.
As a matter of fact, Proposition \ref{corre_thm} can yield  ZCZ sequences which are not obtained in \cite{Popovic_2010_}.
In \cite{Popovic_2010_}, letting $N=tm$  with  positive 	integers $t$  	 and $m$,
 Popovi${\acute{c}}$ and Mauritz proposed GCL--ZCZ sequences
$\textbf{z}=(z(0),z(1),\dots, z(N-1)) $
 satisfying
 	$$z(k) = s(k)\,c(k \ (\mathrm{mod}\,\, m)), \quad k=0,1,\dots,N-1,$$
 	where $\s$ is a ZC sequence and
 	   $\textbf{c}$ is a complex-valued sequence of length $m$.
  Such sequences $\textbf{z}$ have ZCZ length  $t-1$.

Take an  example of $N=35$ with $p=7$ and $q=5$.
Then $\pi(x) = x^7+ax$ with $a\in \{0,5,10,15,25,30\}$.
Let  the ZC sequence  $\s $ be $s(x)=\xi_N^{(x^2+(N\,\mathrm{mod}\, 2 )x)/2}, 0\leq x<N$.
Compared with GCL--ZCZ  sequences $\textbf{z} $, there are two new sequences $\s \circ \pi_1$ and $\s \circ \pi_2$ with $\pi_1=x^7+10x$
\vspace{2mm}
and $\pi_2=x^7+15x$.
For $\s \circ \pi=(\xi_N^{e_0}, \xi_N^{e_1},...,\xi_N^{e_{N-1}})$, define its
exponent sequence
as $\textbf{e}=(e_0,e_1,...,e_{N-1})$.
Then the exponent sequences $\textbf{e}$ and $\textbf{e}'$ of
 $\s \circ \pi_1$ and  $\s \circ \pi_2$, respectively, can be expressed as follows:
 \begin{align*}
\textbf{e}=(&  0, 31,  1, 8, 10, 0, 6,
21, 3, 15, 15, 31, 21, 13, 0, 10, 1, 1, \\
&  3, 0, 20,
21, 31, 8, 15, 10, 21, 6,
28, 10, 15, 1, 31, 28, 20),\\
\vspace{1mm}
\textbf{e}'=(&  0, 31,31, 28, 15,20, 1,
 21, 3, 10, 0, 1, 6, 8,
0, 10, 31,  21,\\
&   8, 20, 15, 21, 31, 3, 0, 15, 6, 1,
28, 10, 10, 21, 1, 13, 15).
\end{align*}

Below, we use the interleaved sequence  $\s \circ \pi_1$ as an example to demonstrate the two observations in the proof sketch of Proposition~\ref{corre_thm}.

For the interleaved sequence  $\s \circ \pi_1$, consider the shift $d=2$,
then
	$\theta(2) =\sum^{N-1}\limits_{k=0} \zeta_{35}^{ h_2(k)}$
with
$$
	\begin{array}{c}
		h_2(k) = \sum\limits_{1\leq i \leq 14} \binom{14}{i}2^{i-1}k^{14-i}   +\sum\limits_{1\leq i \leq 7} \binom{7}{i}2^{i-1}k^{7-i}  + 10 \sum\limits_{1\leq i \leq 8} \binom{8}{i}2^{i}k^{8-i}  +25k,
	\end{array}
	$$
From the expression of $h_2(k)$, we can see that $h_2(k+35) \equiv h_2(k)  \ (\mathrm{mod}\,\, 35)$.
Thus $(h_2(k))_{k\geq 0}$ modulo $35$ is a sequence of period $35$, which is precisely given by
$$
\begin{aligned}
(h_2(k))_{0\leq k<35} = & (1, 12, 9, 27, 31, 21, 32, 29, 12, 16, 6, 17, 14, 32, 1, 26, 2,   \\ & \hspace{1mm} 34, 17, 21, 11, 22, 19, 2, 6, 31, 7, 4, 22, 26, 16, 27, 24, 7, 11).
\end{aligned}
$$
Moreover, for $t = \frac{2N}{\gcd(d,N)p}=10 $, we have
$$
\begin{aligned}
(h_2(k+10))_{0\leq k<35} =&
( 6, 17, 14, 32, 1, 26, 2, 34, 17, 21, 11, 22, 19, 2, 6, 31, 7, \\ & \hspace{1mm}  4, 22, 26, 16, 27, 24, 7, 11, 1, 12, 9, 27, 31, 21, 32, 29, 12, 16),
\end{aligned}
$$ which can be verified by
$$
  h_2(k+10)-h_2(k)=5 \,\,(\mathrm{mod}\,\,\, 35).
$$
Thus
$\sum^{N-1}\limits_{k=0} \zeta_{N}^{ h_2(k)}=\sum^{N-1}\limits_{k=0} \zeta_{N}^{ h_2(k+t)} =  \zeta_N^5 \sum\limits_{k=0}^{N-1} \zeta_N^{h_2(k)} $, implying that
$\theta(2) =\sum^{N-1}\limits_{k=0} \zeta_{N}^{ h_2(k)}=0$.
As a matter of fact, according to the relation $h_2(k+35) \equiv h_2(k)  \ (\mathrm{mod}\,\, 35)$ and $ h_2(k+10)-h_2(k) = 5  \ (\mathrm{mod}\,\, 35)$,
one can re-arrange $(h_2(k))_{0\leq k<35}$ as
$$
\begin{aligned}
(&1, 6, 11, 16, 21, 26, 31, 2, 7, 12, 17, 22, 27, 32,  \\  &4, 9,14, 19, 24, 29, 34, 2, 7, 12, 17, 22, 27, 32, \\ &1, 6, 11, 16, 21, 26, 31),
\end{aligned}
$$
indicating that
$$
\sum^{N-1}\limits_{k=0} \zeta_{N}^{ h_2(k)} = 2 \zeta_{N} \sum_{i=0}^{p-1} \zeta_p^i + 2 \zeta_{N}^2 \sum_{i=0}^{p-1} \zeta_p^i + \zeta_{N}^4 \sum_{i=0}^{p-1} \zeta_p^i = 0,
$$ where $\zeta_p$ is a $p$-th root of unity, which satisfies $\sum_{i=0}^{p-1} \zeta_p^i = 0$.
It is consistent with the proof of Proposition \ref{corre_thm}.
\end{remark}

%\subsection{Permutation polynomials over  $\mathbb{Z}_{N}$}\label{Sec3-1}

Below we give an explicit  class of PPs over $\mathbb{Z}_N$ of the form given in Proposition \ref{corre_thm}.

\begin{lem}\label{lem_2}
Let $p$ be a prime,  $n$ be a positive integer and $N=p^nN_1$ with $\gcd(p,N_1)=1$.
Suppose $N_1=1$ or $N_1=q_1 q_2 \dots q_r$ where $q_1,\dots, q_r$ are different primes and each $q_i$ for $1\leq i\leq r$ satisfies  $\gcd(p,q_i-1)=1$.
Then the polynomial $\pi(x)=x^p+a x+b$ in $\mathbb{Z}_{N}[x]$ permutes $\mathbb{Z}_{N}$	
 if $a \not\equiv 0,-1 \,\,(\mathrm{mod}\,\, p)$ and $N_1\,|\,a$.
In particular, when $p \equiv 1 \, (\mathrm{mod}\,\, q_i -1)$ for each $i$,
the polynomial	$\pi(x)$ permutes $\mathbb{Z}_{N}$ if and only if
	$a \not\equiv 0,-1 \,\,(\mathrm{mod}\,\, p)$ and $a \not\equiv -1 \,\,(\mathrm{mod}\,\, q_i)$ for each $1\leq i\leq r$.
\end{lem}
\begin{proof}
According to Lemma \ref{N=p^n} (ii), $\pi(x)$ is a PP  over  $\mathbb{Z}_{N}$ if and only if  $\pi(x)$ permutes $\mathbb{Z}_{p^n}$ and $\mathbb{Z}_{q_i}$ for $i=1,2,\dots,r$.
	For $x \in \mathbb{Z}_{p}$, $\pi(x)=x^p+a x+b \equiv (a+1)x+b  \,\,(\mathrm{mod}\,\, p)$ permutes $\mathbb{Z}_{p}$
	if and only if $a \not\equiv -1 \,\,(\mathrm{mod}\,\, p)$.
	In addition, for $x \in \mathbb{Z}_{p^n},$ we have $\pi(x)'=px^{p-1}+a  \equiv a\,\,(\mathrm{mod}\,\, p)$.
	Thus
	according to Lemma \ref{N=p^n} (i), $\pi(x)=x^p+ax$ is a PP over $\mathbb{Z}_{p^n}$
	if and only if $a \not\equiv 0,-1 \,\,(\mathrm{mod}\,\, p)$.
%Since $a \not\equiv 0,-1 \,\,(\mathrm{mod}\,\, p)$,
%it follows from (i) that $\pi$ is a PP over $\mathbb{Z}_{p^n}$.
Moreover, for any $i$, let $x\in \mathbb{Z}_{q_i}$, since $q_i|a$ and $\gcd(p,q_i-1)=1$,
the polynomial	$\pi(x)=x^p+a x+b\equiv  x^p+b \,\,(\mathrm{mod}\,\,q_i )$ is a PP over $\mathbb{Z}_{q_i}$.  Combining these, the first statement follows from Lemma \ref{N=p^n} (ii).

In particular, when $p \equiv 1 \, (\mathrm{mod}\,\, q_i -1)$ for all $i$,
 the polynomial $\pi(x) =x^p+a x+b \equiv (a+1)x+b  \,\,(\mathrm{mod}\,\, q_i)$. Thus $\pi$ is a PP over $\mathbb{Z}_{q_i}$ if and only if
$a \not\equiv -1 \,\,(\mathrm{mod}\,\, q_i)$.
And recall that $\pi(x)=x^p+ax$ is a PP over $\mathbb{Z}_{p^n}$
	if and only if $a \not\equiv 0,-1 \,\,(\mathrm{mod}\,\, p)$.
Then the second statement follows.
\end{proof}

\begin{remark}
Under the condition $p \geq 3 $, one can always find an integer
$a$ satisfying
$a \not\equiv 0, 1 \ \, (\mathrm{mod} \ \, p) $
such that the polynomial $\pi(x) = x^p+a x+b\,\,(\mathrm{mod}\,\, N) $ is a PP.
When $n=1$, the condition  $a \not\equiv 0,-1 \,\,(\mathrm{mod}\,\, p)$, for $\pi(x)$ to be a PP, is reduced to  $a \not\equiv -1 \,\,(\mathrm{mod}\,\, p)$.
\end{remark}

Based on Proposition  \ref{corre_thm}, we see that the auto-correlation of permutation-interleaved sequences is zero in many shifts.
Moreover, this construction framework  can generate ZCZ sequences, even  CAZAC sequences.
Then it suffices to find more PPs which can be used in Proposition  \ref{corre_thm}.
Here we propose an open problem and cordially invite interested readers to attack this problem.

\begin{Open problem}
For any positive integer $N$ and prime $p$ with $p\,|\,N$, find more PPs in $\mathbb{Z}_{N}[x]$  having the form as $\pi(x) = x^{p}+a x+b$.
\end{Open problem}

\section{Permutation-interleaved ZC sequences with the ZAC property}\label{Sec_New_4}
For proposed PPs  and any PP in $\mathbb{Z}_{2^n}[x]$,
this section constructs some CAZAC sequences  by using these PPs or their inverses to
interleave ZC sequences.

\begin{thm}\label{mainresult}
		Let $n\geq 2$ be an integer, $p\geq 3$ be a prime and $N=p^nN_1$ with $N_1=1$ or $N=q_1q_2\cdots q_r$ as given in Lemma \ref{lem_2}.
Let $\s$ be a ZC sequence of length $N$ and $\pi(x)=x^p+ax+b$ be a PP over $\mathbb{Z}_N$ in Lemma \ref{lem_2}. Then

	\noindent{\rm (i)}	$\s \circ \pi$ is a CAZAC sequence
	
	\noindent{\rm (ii)}	$\s \circ \pi^{-1} $ is a CAZAC sequence.
\end{thm}

\begin{thm}\label{mainresult2}
	For $N=2^n$ with $n\geq 2$,
let $\s$ be a ZC sequence of length $N$ and  $\pi(x)=a_mx^m + a_{m-1}x^{m-1} + \dots + a_2x^2 + a_1x + a_0$ be a  PP over $\mathbb{Z}_N$ given in Lemma \ref{N=2^n},
then

\noindent{\rm (i)}	$\s \circ \pi$  is a CAZAC sequence.

\noindent{\rm (ii)}	 $\s \circ \pi^{-1}$ is a CAZAC sequence.
\end{thm}

%\begin{remark}
%The interleaved sequence $\mathbf{s} \circ \pi$ in Proposition \ref{corre_thm} is a CAZAC sequence when $N = p^n$ with $n \geq 2$ and $a \not\equiv 0, -1 \pmod{p}$. This constitutes a special case of Theorem \ref{mainresult}, as justified below:
%
%The conclusion follows from combining Proposition \ref{corre_thm} with Lemma \ref{lem_2}. Proposition \ref{corre_thm} states that $\mathbf{s} \circ \pi$ is a CAZAC sequence provided that $\pi(x) = x^{p} + ax + b$ is a permutation polynomial over $\mathbb{Z}_N$. For $n \geq 2$, Lemma \ref{lem_2} establishes that this condition is satisfied if and only if $a \not\equiv 0, -1 \pmod{p}$.
%\end{remark}

The rest of this section is devoted to proving Theorems \ref{mainresult} and \ref{mainresult2}.
 The ZAC properties of interleaved ZC sequences in Theorem \ref{mainresult}  (i) and Theorem \ref{mainresult2} (i) will be proved by an observation that the involved polynomials are PPs; and the ZAC properties of interleaved ZC sequences in Theorem \ref{mainresult} (ii) and Theorem \ref{mainresult2} (ii) will be studied in terms of their DFT sequences. Thus we treat them separately in the subsequent subsections.

\subsection{Proofs of Theorems \ref{mainresult} (i) and \ref{mainresult2} (i) }\label{Sec4_1_new}

 We first present a proposition, which lays a foundation for the proof of Theorem \ref{mainresult} (i).

\begin{prop}\label{prop_PP}
	Let $h_d(k)$ be as in \eqref{h_d(k_form)} and denote
$H_d(k) = h_d(k) - h_d(0)$.
Then under the conditions of Theorem \ref{mainresult}, the  polynomial  $H_d(k)/\gcd(d,N)$ is a PP over $\mathbb{Z}_{N/\gcd(d,N)}$.
\end{prop}
\begin{proof}
It follows that
%	$$ H_d(k)=\frac{1}{2}[(k+d)^{2p}-k^{2p}  +(N\mathrm{mod}\,\,2 )((k+d)^{p}-k^{p}) ] +a[(k+d)^{p+1}-k^{p+1}]  +a^2dk  - M_d $$
 $H_d(0)=0$ for any $d$.
We can see that the integer $\frac{N}{\gcd(d,N)}$ can be decomposed into $\frac{N}{\gcd(d,N)} =\frac{p^n}{\gcd(d,p^n)} \frac{q_1}{\gcd(d,q_1)}
\frac{q_2}{\gcd(d,q_2)} \cdots \frac{q_r}{\gcd(d,q_r)}.$
Then according to Lemma \ref{N=p^n} (ii),
$H_d(k)/ \gcd(d,N)$ is a PP over $\mathbb{Z}_{N/\gcd(d,N) }$ if and only if  $H_d(k)/ \gcd(d,N)$ permutes  $\mathbb{Z}_{{p^n}/{\gcd(d,p^n)}}$ and $\mathbb{Z}_{{q_i}/{\gcd(d,q_i)}}$ for $ i=1,2,\dots,r$.

{\bf Case 1.} We first prove $H_d(k)/ \gcd(d,N)$ permutes $\mathbb{Z}_{{p^n}/{\gcd(d,p^n)}}$.
From Lemma \ref{N=p^n} (i),
it suffices to show that $H_d(k)/ \gcd(d,N)$ permutes $\mathbb{Z}_{p}$ and its derivative $H_d'(k)/ \gcd(d,N) \not\equiv 0  \,\,(\mathrm{mod}\,\,p)$ for any $k \in \mathbb{Z}_{{p^n}/{\gcd(d,p^n)}}$.
In the following, the discussion is divided into two cases: $p\,|\, d$
and $\gcd(p,d)=1$.

{\bf Subcase 1.1.}   $p\,|\, d$.
We first prove that $H_d(k)/ \gcd(d,N)$ is a PP over $\mathbb{Z}_{p}$.
Since
{\small
\begin{align}
 \frac{(k+d)^{h} -k^{h}}{\gcd(d,N)}
&= \frac{d }{\gcd(d,N)}\binom{h}{1}k^{h-1} +d \left[  \frac{d }{\gcd(d,N)}  \sum_{i=2}^{h} \binom{h}{i} k^{h-i} d^{i-2}  \right] \notag \\
& \equiv \frac{d }{\gcd(d,N)} h k^{h-1} \,\,(\mathrm{mod}\,\,p), \notag
\end{align}
}
holds for any positive integer $h$,
we have
{\small
\begin{align}\label{p_zhengchu_d}
\frac{ 2H_d(k) }{\gcd(d,N)}
&  \equiv
\frac{d }{\gcd(d,N)}  \left[ 2p k^{2p-1}
+ (N\mathrm{mod}\,\,2 ) p k^{p-1}
+2a  (p+1) k^{p} + 2a^2k  \right]   \notag \\
& \equiv 2a(a+1)  \frac{d }{\gcd(d,N)}k   \,\,(\mathrm{mod}\,\,p),
\end{align}
}where the last step follows from Euler's theorem.
Due to $\gcd(2,p)=1$, it implies
$\frac{ H_d(k) }{\gcd(d,N)} \equiv a(a+1)  \frac{d }{\gcd(d,N)}k   \,\,(\mathrm{mod}\,\,p),
$
which is a PP over $\mathbb{Z}_{p}$ since  $a,a+1$ and $\frac{d }{\gcd(d,N)}$ are coprime to $p$.

Then we prove that $H_d'(k)/ \gcd(d,N) \not\equiv 0  \,\,(\mathrm{mod}\,\,p)$ for any $k \in \mathbb{Z}_{{p^n}/{\gcd(d,p^n)}}$.
 Similarly as the analysis in \eqref{p_zhengchu_d}, it can be deduced that
$$2H_d'(k)/ \gcd(d,N)
\equiv 2a^2  \frac{ d  }{\gcd(d,N)} \,\,(\mathrm{mod}\,\,p). $$
Due to $\gcd(2,p)=1$,
it follows that  $H_d'(k)/ \gcd(d,N)\,\,(\mathrm{mod}\,\,p)
= a^2  \frac{ d  }{\gcd(d,N)} , $
which is nonzero from  $\gcd(a,p)=1$ and $\gcd(\frac{ d  }{\gcd(d,N)} ,p)=1$.

{\bf Subcase 1.2.}
 $\gcd(p,d)=1$.
We first prove that $H_d(k)/ \gcd(d,N)$ is a PP over $\mathbb{Z}_{p}$.
According to Euler's theorem, it follows
\begin{align}%\label{}
   & 2 H_d(k) \notag \\
 \equiv & [(k+d)^{2}-k^{2}  +d(N\mathrm{mod}\,\,2 ) ] +2a[(k+d)^{2}-k^{2}] +2a^2dk   -2ad^2-d^2-d(N\mathrm{mod}\,\,2)   \notag  \\
\equiv & 2(a+1)^2  dk   \,\,(\mathrm{mod}\,\,p) . \notag
\end{align}
Since
$\gcd(\gcd(d,N),p)=1$,
there exist  positive integers $t$ and $s$ such that
$t\gcd(d,N)+  sp =0$, implying that
$\frac{1 }{\gcd(d,N)} \equiv  t \,\,(\mathrm{mod}\,\,p)$
and $\gcd(t,p)=1$.
Hence,
$$\frac{ 2H_d(k) }{\gcd(d,N)}
\equiv  2t H_d(k)   \,\,(\mathrm{mod}\,\,p) \equiv  2(a+1)^2 t dk   \,\,(\mathrm{mod}\,\,p). $$
From $\gcd(p,2)=1$, it follows that
$
\frac{ H_d(k) }{\gcd(d,N)} \equiv  (a+1)^2 t dk   \,\,(\mathrm{mod}\,\,p),
$
which is a PP over $\mathbb{Z}_{p}$ since
$\gcd(p,a+1)=1$, $\gcd(p,t)=1$ and  $\gcd(p,d)=1$.

Now we prove that $H_d'(k)/ \gcd(d,N) \not\equiv 0  \,\,(\mathrm{mod}\,\,p)$ for any $k \in \mathbb{Z}_{{p^n}/{\gcd(d,p^n)}}$.
 Similarly as the analysis in Subcase 1.1, it derives
$2H_d'(k)
\equiv 2a(a+1) d \,\,(\mathrm{mod}\,\,p). $
Thus
$$\frac{ 2H_d'(k) }{\gcd(d,N)}
\equiv  2t H_d'(k)   \,\,(\mathrm{mod}\,\,p) \equiv   2a(a+1) td   \,\,(\mathrm{mod}\,\,p). $$
Since $2,a,a+1,t$ and $d$ are coprime to $p$, we have
$\frac{ H_d'(k) }{\gcd(d,N)}
\equiv  a(a+1) td   \,\,(\mathrm{mod}\,\,p), $
which is nonzero.
That is $\frac{  H_d'(k) }{\gcd(d,N)}   \not\equiv 0 \,\,(\mathrm{mod}\,\,p)$.

{\bf Case 2.} We prove that $H_d(k)/ \gcd(d,N)$ is a PP over $\mathbb{Z}_{{q_i}/{\gcd(d,q_i)}}, i=1,2,\dots,r$.
For each $i$,
it is obvious for the case of $\gcd(d,q_i)=q_i$,
then
it suffices to consider the case of  $\gcd(d,q_i)=1$, implying that $\gcd( \gcd(d,N),q_i)=1$.
Thus for each $i$, there exists a positive integer $t_i$ with  $\gcd(t_i,q_i)=1$ such that
$\frac{1 }{\gcd(d,N)} \equiv  t_i \,\,(\mathrm{mod}\,\,q_i)$.
In the following, the discussion proceeds in two cases, depending on whether
$q_i \neq 2$ for all $i$, or $q_i = 2$ for  a certain $i$.

{\bf Subcase 2.1.}  $q_i \neq 2$ for any $i=1,2,\dots,r$.
Since $p-1 \equiv 0 \, (\mathrm{mod}\,\, q_i -1)$ for all $1 \leq i \leq r$,
we have $x^p \equiv x \, (\mathrm{mod}\,\, q_i)$ for $x\in \mathbb{Z}_{q_i}$.
Then it follows
$$
  \frac{2H_d(k)}{\gcd(d,N)} \equiv  2 t_iH_d(k) \equiv  2 t_i(a+1)^2 dk   \,\,(\mathrm{mod}\,\,q_i).
$$
 By   $\gcd(2,q_i)=1$, it derives  that
$
   \frac{H_d(k)}{\gcd(d,N)} \equiv  t_i(a+1)^2d k   \,\,(\mathrm{mod}\,\,q_i).
$
Thus  $\frac{H_d(k)}{\gcd(d,N)}$   is a PP over $\mathbb{Z}_{q_i}$ since $\gcd(t_i, q_i)=1 $,
$\gcd(a+1, q_i)=1 $ and $\gcd(d, q_i)=1 $.

{\bf Subcase 2.2.}  There exists $q_i = 2$ for a certain $i$.
Then we have odd $d$, odd $t_i$, even $N$, and even $a$.
Since $\frac{1 }{\gcd(d,N)} \equiv  t_i \,\,(\mathrm{mod}\,\,q_i)$ with odd $t_i$, it follows
 \begin{equation}\label{Hdk}
 \frac{H_d(k)}{\gcd(d,N)} \equiv   t_iH_d(k)  \equiv H_d(k)   \,\,(\mathrm{mod}\,\,2).
 \end{equation}
Note that $H_d(0) =0$
and $H_d(1) \equiv \frac{ 1+d^{2p}- (1+d)^{2p}  }{2}  \,\,(\mathrm{mod}\,\,2) $ by even $N$ and even $a$.
Let $d=2d_1+d_0$ with $d_0=1$. Then
$$
H_d(1)
\equiv \frac{ 1+d_0^{2p}- (1+d_0)^{2p}  }{2}  \,\,(\mathrm{mod}\,\,2)
\equiv \frac{ 1+1- 2^{2p}  }{2}  \equiv 1\,\,(\mathrm{mod}\,\,2) .
$$
Together with \eqref{Hdk}, it is clear that $\frac{H_d(k)}{\gcd(d,N)}$ is a PP over  $\mathbb{Z}_{2}$.

Combining the above cases, it implies that $H_d(k)/ \gcd(d,N)$ is a PP over $\mathbb{Z}_{N/ \gcd(d,N)}$. The desired conclusion thus follows.
\end{proof}

\noindent\textbf{Proof of Theorem \ref{mainresult} (i).}
Recall that $		\theta(d) =C_0\sum^{N-1}\limits_{k=0} \zeta_N^{ h_d(k)}$ in \eqref{corehk},
and  $ H_d(k) = h_d(k) -h_d(0)$ in Proposition \ref{prop_PP}.
% where $h_d(k)$ is given in \eqref{h_d(k_form)} and $M_d = ad^{p+1}+ \frac{d^{2p}+(N\mathrm{mod}\,\,2 )d^{p}}{2}$.
Then we have
$$ \theta(d)=C_0\sum^{N-1}\limits_{k=0} \zeta_N^{ h_d(k)} = C_0 \zeta_N^{ h_d(0)} \sum\limits_{k=0}^{N-1} \zeta_N^{ H_d(k)}
= C_0 \zeta_N^{ h_d(0)} \sum\limits_{k=0}^{N-1} \zeta_{N/ \gcd(d,N)}^{ H_d(k)/ \gcd(d,N)}. $$
Since Proposition \ref{prop_PP} implies that $H_d(k)/\gcd(d,N)$ is a PP over $\mathbb{Z}_{N/\gcd(d,N)}$,
 it follows that $\theta(d) = 0$.
\hfill $\square$

\noindent\textbf{Proof of Theorem \ref{mainresult2} (i).}
Similarly as Proposition \ref{prop_PP} and Theorem \ref{mainresult} (i),   %to the approach %used in %Theorem \ref{mainresult} (i),
Theorem \ref{mainresult2} (i) can be proved.
For the sake of completeness of this paper, we present the detailed proof in
\hyperref[appendix:C]{Appendix C}.
\hfill $\square$

\begin{remark}
For the QPP-interleaved ZC sequences in \cite{Berggren_2024_},
its periodic auto-correlation is given in  Equations (10)-(14) in \cite{Berggren_2024_}, as
$\theta(d)=C \sum_{k=0}^{N-1} \xi_{N}^{ g_3k^3+g_2k^2+g_1k  } $
with constant $C$ and $d\,|\, g_i, i=1,2,3$.
The proof of the  ZAC property in~\cite[Thm.~1]{Berggren_2024_}
can be reduced to showing  that  the
cubic polynomials $\frac{g_3k^3+g_2k^2+g_1k}{\gcd(d,N)} $ are PPs over $\mathbb{Z}_{N/ \gcd(d,N)}$.
Moreover, a simple coefficient test for
cubic PPs over an  integer residue ring has been proposed in \cite{Chen_2006_, Zhao_Fan}, which can
be used to directly verify that $\frac{g_3k^3+g_2k^2+g_1k}{\gcd(d,N)} $ is a cubic PP over $\mathbb{Z}_{N/ \gcd(d,N)}$.
\end{remark}

\nn
Below we provide an example to illustrate the results of Theorems \ref{mainresult} (i) and \ref{mainresult2} (i).

\begin{example}
	When $N=9=3^2$,	let $u=1$ and $l=0$, then the
	ZC sequence $\s$ satisfies $s(k)=\xi_9^{(k^2+k)/2}, k=0,1,\dots,8.$
%	That is $\s=(1,\xi_9,\xi_9^3,\xi_9^6,\xi_9,\xi_9^6,\xi_9^3,\xi_9,1)$.
	Let the PP $\pi(x) = x^3+x \,(\mathrm{mod}\,\, 9)$, then $(s \circ \pi)(k) =s(\pi(k)) =\xi_9^{(\pi(k)^2+\pi(k))/2}, k=0,1,\dots,8$.
	Thus the interleaved ZC sequence is given by  $\textbf{y}= \s \circ \pi=(1,\xi_9^3,\xi_9,\xi_9^6,\xi_9^6,
	\xi_9,\xi_9^3,1,\xi_9).$
By calculating the periodic auto-correlation  of $ \textbf{y}$, we obtain $\theta(d)=0$ for $0<d<N$,
which  is consistent with the result in Theorem \ref{mainresult} (i).

When $N=16=2^4$, let
$u=1$ and $l=0$, then the ZC sequence $\s$ satisfies
$s(k)=\xi_{16}^{k^2/2}, k=0,1,\dots,15.$
When $l=4$, the fourth-degree polynomial
 $\pi(x) = x^4+x^2+x\,(\mathrm{mod}\,\,16)$
 is a PP over $\mathbb{Z}_{16}$ by Lemma \ref{N=2^n}.
Thus, the interleaved sequence $\textbf{y}= \s \circ \pi$ satisfies
$y(k) =\xi_{16}^{(k^4+k^2+k)^2/2}, k=0,1,\dots,15$.
A direct computation shows that the periodic auto-correlation  $\theta(d)$ of $ \textbf{y}$ equals zero for all nonzero delays
$d$.
This result aligns with Theorem \ref{mainresult2} (i).

\end{example}

%\subsection{The inverse of permutation polynomials interleaved ZC sequences}\label{Sec3-3}	
\subsection{Proofs of Theorems \ref{mainresult} (ii) and \ref{mainresult2} (ii) }\label{Sec3-3}

The  proof of Theorem \ref{mainresult} (ii) employs the Fourier transform to demonstrate that the inverses of PPs interleaved sequences yield  constant amplitude sequences in the Fourier domain.

\begin{lem}(\!\!\cite{Trends}) \label{CA_ZAC}
	Given a complex sequence $\s=(s(0),s(1),\dots,s(N-1))$,
	its discrete Fourier transform (DFT) is defined by
	$$\widehat{s}(m) = \frac{1}{\sqrt{N}}\sum_{k=0}^{N-1} s(k) \xi_N^{mk}, \,m= 0, 1, \dots, N- 1.$$
	Then
	$\s=(s(0),s(1),\dots,s(N-1))$ is a ZAC sequence if and only if its DFT sequence
	$\widehat{\s}=(\widehat{s}(0),\widehat{s}(1),\dots,\widehat{s}(N-1))$  	
	 is a constant amplitude (CA) sequence.
\end{lem}

%Since $\textbf{y}=\s \circ \pi^{-1}$ is a CA sequence, it suffices to show its ZAC property.
According to Lemma \ref{CA_ZAC},
for an interleaved sequence $\y= \s \circ \pi^{-1}$,
 we only need  to prove that
its DFT sequence
	$\widehat{\y}=(\widehat{y}(0),\widehat{y}(1),\dots,\widehat{y}(N-1))$ is a CA sequence,
where	
$$\widehat{y}(m)
=\frac{1}{\sqrt{N}}\sum\limits_{k=0}^{N-1} y(k) \xi_N^{mk}
=\frac{1}{\sqrt{N}}\sum\limits_{k=0}^{N-1} s(\pi^{-1}(k)) \xi_N^{mk}, \,m= 0, 1, \dots, N- 1.$$
It follows that
\begin{align}\label{xx*first}
		& \quad \  \widehat{y}(m)\widehat{y}^*(m) \notag \\
		&\stackrel{}{=}  \frac{1}{\sqrt{N}} \sum_{k=0}^{N-1}s(\pi^{-1}(k)) \xi_N^{km} \, \frac{1}{\sqrt{N} }\sum_{h=0}^{N-1}s^* (\pi^{-1}(h)) \xi_N^{-hm}  \notag \\
		&\stackrel{\mathrm{(a)}}{=} \,	 \frac{1}{N} \sum_{k=0}^{N-1}s(k) \xi_N^{\pi (k)m} \,  \sum_{h=0}^{N-1}s^* (h) \xi_N^{-\pi (h)m} \notag\\
		&\stackrel{\mathrm{(b)}}{=} \, \frac{1}{N} \sum_{k=0}^{N-1}s(k) \xi_N^{\pi (k)m} \,  \sum_{d=0}^{N-1}s^* ((k+d)\, (\mathrm{mod}\,N) ) \xi_N^{-\pi (k+d)m} \notag \\
		&\stackrel{}{=} \, \frac{\xi_N^{b m}}{N}\,  \sum_{d=0}^{N-1}\xi_N^{-\pi (d)m}
		\,  \sum_{k=0}^{N-1}s(k)s^* (k+d) \xi_N^{-(\pi (k+d)-\pi(k)-\pi(d) +b	)m}  \notag\\
		&\stackrel{\mathrm{(c)}}{=} \, \frac{\xi_N^{b m}}{N} \, \sum_{d=0}^{N-1}\xi_N^{-\pi (d)m
		-ud( d+ (N  \, \mathrm{mod}\,\,2)+2l)/2}
	\,  \sum_{k=0}^{N-1}\xi_N^{- udk-(\pi(k+d)-\pi(k)-\pi(d) +b	)m 	} \notag\\
	&\stackrel{}{=} \, \frac{\xi_N^{b m}}{N} \, \sum_{d=0}^{N-1}\xi_N^{-\pi (d)m
		-ud( d+ (N  \, \mathrm{mod}\,\,2)+2l)/2}
	\, M(d)
\end{align}
where
%Let
\begin{equation}\label{MMMvv}
	M(d)=\sum\limits_{k=0}^{N-1}\xi_N^{- udk-(\pi (k+d)-\pi(k)-\pi(d) +b	)m },
\end{equation}
step (a) follows from $\pi^{-1} (\pi(k))=k$;
since $\pi(t+N )\equiv  \pi(t ) \,\mathrm{mod}\,\, N $ for any $t$,
step (b) is obtained by  introducing the variable $h=k+d$, which reorders the terms in the sum;
step (c) holds by $s(k)=\xi_N^{u(k^2+(N\,\mathrm{mod}\,\, 2)k+2lk)/2}$.
Below we shall show that $ \widehat{y}(m)\widehat{y}^*(m)$ in \eqref{xx*first} is a constant for any $0\leq m<N$.

\noindent\textbf{Proof of Theorem \ref{mainresult} (ii).}
%According to Lemma \ref{CA_ZAC}, it suffices to show that
%$ \widehat{y}(m)\widehat{y}^*(m)$ in \eqref{xx*first} is a constant for any $0\leq m<N$.
By substituting the PP $\pi(x)=x^p+a x+b$
 into $M(d)$ in \eqref{MMMvv}, it  yields
\begin{equation}\label{inner_sum}
	M(d)= \sum_{k=0}^{N-1}\xi_N^{- udk-(  (k+d)^p-k^p-d^p 	)m }= \sum_{k=0}^{N-1}\xi_N^{- h_d(k)}=\sum_{k=0}^{N-1}\xi_{N/ \gcd(d,N)}^{- h_d(k)/ \gcd(d,N)},	
\end{equation}
where $\gcd(u,N)=1$, $d \in \mathbb{Z}_N$, $m \in \mathbb{Z}_N$ and
$$h_d(k)= (  (k+d)^p-k^p-d^p 	)m +udk . $$
We now discuss $M(d)$ in two cases according to the values of $d$.
When $d=0$, it is evident that $M(d)=N$.
When $d\neq 0$, according to Lemma \ref{Apend_thm2_PP} in
\hyperref[appendix:D]{Appendix D}, we know that $h_d(k)/ \gcd(d,N)$ is a PP over $\mathbb{Z}_{N/ \gcd(d,N)}$.
Consequently, it follows from \eqref{inner_sum} that
 \begin{equation*}%\label{sumsum}
	M(d)
	=	\left\{\begin{array}{cll}
		0 , &   \text{ if } d \neq 0, \\
		N,   &  \text{ if } d=0. \\
	\end{array}\right.	
\end{equation*}
Therefore, we have
%\begin{center}
\begin{align*}%\label{xx*}
  \widehat{y}(m)\widehat{y}^*(m)
	&{=}\frac{\xi_N^{bm}}{N} \, \sum_{d=0}^{N-1}\xi_N^{-\pi^{-1} (d)m
-ud( d+ (N  \, \mathrm{mod}\,\,2)+2l)/2}
	 \, M(d) \notag
	{=}\xi_N^{b m}\xi_N^{- \pi^{-1} (0)m}
	{=}1, \notag
\end{align*}
%\end{center}
where the last step holds by $\pi^{-1} (0)=b $.
Hence, the proof is finished.
\hfill $\square$

\noindent\textbf{Proof of Theorem \ref{mainresult2} (ii).}
The proof of Theorem \ref{mainresult2} (ii) is similar
to that of  Theorem \ref{mainresult} (ii). For the sake of completeness of this paper, we present the proof in
\hyperref[appendix:E]{Appendix E}.
\hfill $\square$

\begin{example}
Let $N=27=3^3$, a PP $\pi(x) = x^3+x$ over $\mathbb{Z}_N$	induces a permutation $(0, 2, 10, 3, 14, 22, 6, 26, 7, 9, 11, 19, 12, 23, 4, 15, 8, 16, 18, 20, 1, 21, 5, 13, 24, 17, 25)$ of $\mathbb{Z}_N$,
thus its inverse $\pi^{-1}$ is given by $(0, 20, 1, 3, 14, 22, 6, 8, 16, 9, 2, 10, 12, 23, 4, 15, 17, 25, 18, \\ 11,
 19, 21, 5, 13, 24, 26, 7)$.
Then for a ZC sequence $\s$ with $s(x)=\xi_N^{(x^2+x)/2}$, the interleaved sequence
is given by
$\s \circ \pi^{-1}= ( \xi_{N}^{ e_0},\xi_{N}^{ e_1},...,\xi_{N}^{ e_{N-1}}) $,
where its exponent sequence $\textbf{e}$  is
$$
 \textbf{e}=(0,21,1,6,24,10,21,
9,1,18,3,1,24,6, 10,12,18,1,9,12,1,
15,15,10,3,0,1).
$$
We can calculate that the auto-correlation of $\s \circ \pi^{-1}$ equals zero for all nonzero delays.	
It is consistent with the result in Theorem \ref{mainresult} (ii).
\end{example}

%\section{Comparison with known results}\label{Sec4}
\section{Inequivalent permutation-interleaved ZC sequences}\label{Sec4}
The previous section gives a class of PPs $\pi(x)$ from which many permutation-interleaved ZC sequences $\s\circ \pi$ exhibit the ZAC property.
A natural question arises here: are those permutation-interleaved ZC sequences covered by the equivalence classes of ZC sequences?
This section  addresses  this question.
We will first prove the sufficiency of the conjecture given in \cite{Berggren_2024_} and generalize the result to $\s \circ \pi^{-1}$ for QPPs $\pi(x)$.
Moreover, we show that some proposed CAZAC sequences are  neither covered by
the equivalence class of ZC sequences,  nor  covered by
the equivalence class of %QPP-interleaved ZC sequences.
interleaved ZC sequences by QPPs and their inverses.

Before  presenting  the results, we give an auxiliary lemma.

 \begin{lem}(\!\!\cite{Rosen} Theorem 4.20) \label{solution_lem}
 	Assume $A$ is an $n\times n$ matrix of integers, $X$ is a vector of $n$ unknowns, and $B$ is a vector of $n$ integers.
 	If $\gcd(|A|,m)=1$, then $AX \equiv B \,\,(\mathrm{mod}\,\, m)$ has a unique set of incongruent solutions
 	$X \equiv \overline{\Delta}(adj(A))B \,\,(\mathrm{mod}\,\, m)$, where $|A|$ is the determinant of $A$, $\Delta \equiv |A|  \,(\mathrm{mod}\,\, m)$, $\overline{\Delta}\Delta   \equiv 1 \,(\mathrm{mod}\,\, m)$ and $adj(A)$ is the adjoint of $A$.
 \end{lem}

Recall the conjecture about inequivalent QPP-interleaved ZC sequences
in Conjecture \ref{conj_Ber}.
Based on Lemma \ref{solution_lem},
we shall give  an affirmative answer to the conjecture on its sufficiency,
% a positive answer of  its sufficiency,
and prove that ZC sequences interleaved by QPPs' inverses  are also inequivalent to ZC sequences.

\begin{prop}\label{prop1}
	Let $\s$ be ZC sequences of length $N$ and $\pi=f_2x^2+f_1x +f_0$ be QPPs over $\mathbb{Z}_N$.
	When $N=p^n$ with $n\geq 2$ and a prime $p\geq 5$,
CAZAC sequences $\s \circ \pi$ and $\s \circ \pi^{-1}$  are both inequivalent
to ZC sequences.
%That is, they cannot be generated by
%compositions of the five operations in \eqref{mathe_opera} on ZC sequences.
\end{prop}
\begin{proof}
	Suppose that $\s \circ \pi$ is   equivalent to ZC sequences.  Then  from \eqref{uuuuu} there exist parameters
 $v,d\in \mathbb{Z}_N$, 	$r\in \mathbb{Z}_{2N}$,  $s\in \{-1,1\}$ and $\gcd(u_i,N)=1$ with $i=1,2,$
such that	
$$
\xi_N^{u_1(\pi(k)+(N\,\mathrm{mod}\,2))\pi(k)/2} = \xi_N^{vk}\xi_N^{su_2(k+d+(N\,\mathrm{mod}\,2))(k+d)/2+r/2}
$$	
for any $k=0,1,2,\dots,N-1$.
Thus
$$u_1(\pi(k)+1)\pi(k) - 2vk - su_2(k+d+1)(k+d)-r \equiv 0
\,\,(\mathrm{mod}\,\,2p^n)
 $$
 for any $k\in \mathbb{Z}_{N}$.
Since $\gcd(2,p^n)=1$, we have
\begin{equation}\label{BBB}
	b_4 k^4 +b_3 k^3 + b_2k^2 +b_1k  +b_0 \equiv 0 \,\,(\mathrm{mod}\,\,p^n),
\end{equation}
for any $k\in \mathbb{Z}_{p^n}$,
where $b_4 = u_1 f_2^2$, $b_3 = 2u_1 f_2f_1$, $b_2 = u_1(f_1^2+f_2 +2f_2f_0)-su_2 $, $b_1 = u_1f_1(2f_0+1)-2v-su_2(2d+1) $ and $b_0 = u_1f_0(f_0+1)-su_2(d^2+d) -r$.

Since $p\geq 5$,
we can rewrite equation \eqref{BBB} for
$k=0,1,2,3,4 \in \mathbb{Z}_{p^n}$ as
\begin{eqnarray}\label{Eq-matrix}
	B	\left(\begin{array}{c}
		b_{4}	\\
		b_{3}	\\
		b_{2}	\\
		b_{1}	\\
		b_{0}	\\	
	\end{array}	\right) =	\left(
	\begin{array}{ccccc}	
		0^4 & 0^{3} & 0^{2} & 0 & 1 \\
		1^{4} & 1^{3} & 1^{2} & 1 & 1 \\
		2^{4} & 2^{3} & 2^{2} & 2 & 1 \\
		3^{4} & 3^{3} & 3^{2} & 3 & 1 \\
		4^{4} & 4^{3} & 4^{2} & 4 & 1 \\
	\end{array}	\right)
	\left(\begin{array}{c}
		b_{4}	\\
b_{3}	\\
b_{2}	\\
b_{1}	\\
b_{0}	\\		
	\end{array}	\right)
	\equiv \left(
	\begin{array}{c}
		0\\
		0 \\
		0	\\
		0 \\
		0 \\
	\end{array}
	\right)
	(\mathrm{mod}\,\, p^n).	
\end{eqnarray}
Since the above matrix $B$ is a Vandermonde matrix, it is clear that $|B|=\prod_{0\leq i<j \leq 4}(j-i)$, which implies that  $\gcd(|B|,p^n)=1$ by a  prime $p \geq 5$.
Based on Lemma \ref{solution_lem} and $\gcd(|B|,p^n)=1$, \eqref{Eq-matrix} has a unique set of incongruent solutions
$b_i  \equiv 0  \,(\mathrm{mod}\,\, p^n), 0\leq i\leq p-1$.
For $b_3 = 2u_1 f_2f_1$, since $2, u_1$ and $f_1$ are coprime to $p$, it follows that $f_2 \equiv 0 \,\,	(\mathrm{mod}\,\, p^n),$
a contradiction.
Thus the assumption does not hold. Therefore	QPP-interleaved sequences $\s \circ \pi$ are inequivalent to ZC sequences.

Suppose that $\s \circ \pi^{-1}$ is equivalent to ZC sequences, then
$$u_1(\pi^{-1}(k)+1)\pi^{-1}(k) - 2vk - su_2(k+d+1)(k+d)-r \equiv 0
\,\,(\mathrm{mod}\,\,2p^n).
$$
Thus from $\pi^{-1}(\pi(k))=k$ and $\gcd(2,p)=1$, we have
$$u_1(k+1)k - 2v\pi(k) - su_2(\pi(k)+d+1)(\pi(k)+d)-r \equiv 0
\,\,(\mathrm{mod}\,\,p^n).
$$
Similar to the proof of $\s \circ \pi$, we can get a contradiction, implying that
 $\s \circ \pi^{-1}$ is also inequivalent to ZC sequences.
\end{proof}

We now consider the equivalence class of
 QPP-interleaved ZC sequences as the set of all sequences
generated by composing the five operations in \eqref{mathe_opera} on QPP-interleaved ZC sequences.
And the equivalence class of QPPs' inverse-interleaved ZC sequences can also be defined similarly.

In the following,
we show that %when $N=p^nq_1q_2\dots q_r$ with $n\geq 2$,
all interleaved CAZAC sequences given in Theorem \ref{mainresult}  are inequivalent to ZC sequences,
and some are inequivalent to QPPs and their inverses interleaved ZC sequences
\cite{Berggren_2024_}.
%cannot be covered by the equivalence class of QPPs and their inverses interleaved ZC sequences
Note that for QPPs $f_2x^2+f_1x$ over $\mathbb{Z}_N$, when $f_2=0$,
 $f_1x$ corresponds the operation of Decimation  in \eqref{mathe_opera},
 thus
the equivalence class of QPP-interleaved ZC sequences is reduced to
the equivalence class of ZC sequences.
Moreover since $\pi^{-1}(\pi(k))=k$,
the case of $\s \circ \pi^{-1}$ can be proved  similarly  by  substituting $k$ into $\pi(k)$.
Therefore, we only give the detailed proof of Proposition \ref{pro2} (ii).

\begin{prop}\label{pro2}
 Let  $p\geq 3$ be a prime, $n$ be a positive integer and $N=p^nN_1$ as in Theorem~\ref{mainresult}.
  Then for ZC sequences $\s$ and PPs $\pi(x) = x^p+ax +b$ in $\mathbb{Z}_N[x]$ with  $a \not\equiv 0,-1 \,\, (\mathrm{mod}\,\, p)$,
  %Let $\s$ be ZC sequences and $\pi(x) = x^p+ax$ with $a \not\equiv 0,-1 \,\, (\mathrm{mod}\,\, p)$ be PPs over $\mathbb{Z}_N$.
we have	\\
\noindent{\rm (i)} CAZAC sequences $\s \circ \pi$ and  $\s \circ \pi^{-1}$ are inequivalent to ZC sequences, respectively.

\noindent{\rm (ii)} CAZAC sequences $\s \circ \pi$ are inequivalent to
interleaved ZC sequences by QPPs and their inverses.
\end{prop}
\begin{proof}
%(i) The proof of (i)  follows immediately from Proposition \ref{prop1} and the proof method in case (ii).
(ii)	Suppose that $\s \circ \pi$ are equivalent to
interleaved ZC sequences by QPPs or their inverses,
%can be
%generated by
%compositions of operations in \eqref{mathe_opera} and QPPs or QPPs' inverses on ZC sequences,
then from \eqref{uuuuu} there exist parameters
$v,d\in \mathbb{Z}_N$, 	$r\in \mathbb{Z}_{2N}$,  $s\in \{-1,1\}$ and $\gcd(u_i,N)=1$ with $i=1,2,$
 such that	
\begin{equation*}%\label{com_PP}
	\xi_N^{u_1(\pi(k)+(N\,\mathrm{mod}\,2))\pi(k)/2} = \xi_N^{vk}\xi_N^{su_2(\pi'(k)+d+(N\,\mathrm{mod}\,2))(\pi'(k)+d)/2+r/2}
\text{ for any } k\in \mathbb{Z}_N,
\end{equation*}	
or
\begin{equation}\label{com_PPni}
	\xi_N^{u_1(\pi(k)+(N\,\mathrm{mod}\,2))\pi(k)/2} = \xi_N^{vk}\xi_N^{su_2(\pi'(k)^{-1}+d+(N\,\mathrm{mod}\,2))(\pi'(k)^{-1}+d)/2+r/2}
\text{ for any } k\in \mathbb{Z}_N,
\end{equation}	
where QPPs $ \pi'(k)=f_2k^2+f_1k$ with $f_2 \equiv0 \,(\mathrm{mod}\,\,p)$ and $f_1 \not\equiv 0 \,(\mathrm{mod}\,\,p).$

This implies the exponent polynomial
$	
F(k) \equiv 0 \,\,(\mathrm{mod}\,\,2N) \text{ for any }  k\in \mathbb{Z}_{N},
$
where
\begin{equation}\label{ex_F_case1}
F(k)=	u_1(\pi(k)+(N\,\mathrm{mod}\,2))\pi(k) - 2vk - su_2(f_2k^2+f_1k+d+1)(f_2k^2+f_1k+d)-r
\end{equation}
or
\begin{align}\label{ex_F_case2}
F(k)=  & \,u_1(\pi'(k)^p+a\pi'(k)+b+(N\,\mathrm{mod}\,2))(\pi'(k)^p+a\pi'(k)+b) \notag \\
   &  \,- 2v\pi'(k) - su_2(k+d+(N\,\mathrm{mod}\,2))(k+d)-r
\end{align}
by setting $k=\pi'(k)$ in \eqref{com_PPni}.

We now discuss the cases in which either equality \eqref{ex_F_case1} or equality \eqref{ex_F_case2} is satisfied.
For each case, we  first study  $F(k)\equiv  0 \,\,(\mathrm{mod}\,\,p)$ by Lemma \ref{solution_lem},  then investigate the two first-order finite difference $F(k+p)-F(k)\equiv  0 \,\,(\mathrm{mod}\,\,p^2)$ for $k=0,1$.
Combining these results leads to a contradiction, under the conditions $(u_1, p) = 1$ and $a \not\equiv 0, -1 \pmod{p}$.
The detailed proofs for the two cases are provided in
\hyperref[appendix:F]{Appendix F}.
\end{proof}

\begin{remark}
 Berggren and Popovi$ {\acute{c}}$ showed that for $N\in \{9,18,36,45,63,90,99,117,126\}$, there is no  inequivalent CAZAC  sequence for
QPP-interleaved sequences\cite{Berggren_2024_}. Our results can give  inequivalent  CAZAC sequences for $N\in \{9,18,45,63,90,99,117,126\}$, only missing $N=36$.
%For the case of $N=2^n$, there also exist unique CAZAC sequences proposed in this paper.
%Moreover, certain PPs' inverses  interleaved ZC sequences proposed in this paper are also  inequivalent.
\end{remark}

 In addition, for the case of $N=2^n$ in Theorem \ref{mainresult2}, there also exist interleaved CAZAC sequences inequivalent to those from QPPs and their inverses.
%Here, some examples are given to illustrate the results in this section.
In Example \ref{someex} (i), we show that half of CAZAC sequences generated by Theorem   \ref{mainresult2} (i) are inequivalent to those from QPPs and their inverses,    when $N=2^4$ and the PP has degree four.

%Moreover, we give some examples to illustrate Proposition 2.

\begin{example}\label{someex}
(i)
	When $N=16=2^4$,
%the root index satisfies $2 \nmid u_1$ and
 all PPs $\pi= a_4x^4 + a_3x^3 + a_2x^2 + a_1x$ over $\mathbb{Z}_N$ satisfy $2 \nmid a_1$, $2\,|\,a_3$ and $2\,|\, a_2+a_4$.
	 The coefficients of QPPs $f_2x^2+f_1x$  over $\mathbb{Z}_N$  satisfy  $2\,|\, f_2$ and $2 \nmid f_1$.
Consider the root index  $u=1$ and $\pi_1=x^4+x^2+x \,(\mathrm{mod}\,\, 16)$,
 then the interleaved  ZC sequence is $\s\circ \pi_1 =(\xi_{2N}^{e_0}, \xi_{2N}^{e_1}, \dots, \xi_{2N}^{e_{N-1}} ) $ with the exponent sequence
$$\textbf{e} = (e_0,e_1,\dots, e_{N-1}) = (0, 9, 4, 9, 16, 1, 4, 17, 0, 25, 4, 25,  16, 17, 4, 1         ).$$
We can verify that
	$\s\circ \pi_1$
is inequivalent to QPPs and their inverses interleaved ZC
sequences.

Moreover, when $N=16=2^4$, experiments show that 128  CAZAC sequences are obtained from all root exponents
$u$ and fourth-degree PPs $\pi$. There are no duplicate sequences among the 128 sequences. Specifically, when $u\in \{1,3,5,7\}$
 and the coefficients of
$\pi= a_4x^4 + a_3x^3 + a_2x^2 + a_1x$
 take certain values, one inequivalent CAZAC sequence is generated for each combination.
 %As listed in Table \ref{tab_3} for $u=1$,
Each root index yields 16 inequivalent sequences for the same set of PPs
$\pi$.
Table \ref{tab_3} lists inequivalent sequences with the root index $u=1$.
Considering all four root indices, we obtain a total of $4\times16=64$ inequivalent sequences,
yielding the proportion of inequivalence sequences  of $64/128=1/2.$
%Thus, the total across all four root indices is
%$4\times16=64$ inequivalent sequences,
%resulting in an inequivalent sequence probability of

	\begin{table}[!ht]
		%\scriptsize
		\caption{Inequivalent   CAZAC sequences $(\xi_{32}^{e_0},\xi_{32}^{e_1},\xi_{32}^{e_2}, \dots,\xi_{32}^{e_{15}})$ constructed by $\s \circ \pi$ with $\pi= x^4 + a_2x^2 + a_1x$ and the root index $u=1$. }	\label{tab_3}
		\begin{center}
			\begin{tabular}{|c|c|c|c|c| p{0.5cm}|}
				\hline
The coefficients of $\pi$ & Exponent sequences
$\textbf{e}=(e_0,e_1,\dots,e_{15})  $
 of  inequivalent    \\[4pt]
	$(a_2,a_1) $ & CAZAC sequences  $\s \circ \pi= (\xi_{32}^{e_0},\xi_{32}^{e_1},\xi_{32}^{e_2}, \dots,\xi_{32}^{e_{15}})$   \\[4pt]
				\hline
				\hline
$(1,1)$  &
				(0, 9, 4, 9, 16, 1, 4, 17, 0, 25, 4, 25, 16, 17, 4, 1)
				   \\[3pt]
		 $(1,3)$   &  (0, 25, 4, 9, 16, 17, 4, 17, 0, 9, 4, 25, 16, 1, 4, 1) \\[3pt]
		 $(1,5)$    &  (0, 17, 4, 17, 16, 9, 4, 25, 0, 1, 4, 1, 16, 25, 4, 9) \\[3pt]
			$(1,7)$  &   (0, 17, 4, 1, 16, 9, 4, 9, 0, 1, 4, 17, 16, 25, 4, 25) \\[3pt]
			$(1,9)$  &   (0, 25, 4, 25, 16, 17, 4, 1, 0, 9, 4, 9, 16, 1, 4, 17)    \\[3pt]
				$(1,11)$  &  (0, 9, 4, 25, 16, 1, 4, 1, 0, 25, 4, 9, 16, 17, 4, 17)   \\[3pt]
				$(1,13)$    &	(0, 1, 4, 1, 16, 25, 4, 9, 0, 17, 4, 17, 16, 9, 4, 25) \\[3pt]
	$(1,15)$   &  (0, 1, 4, 17, 16, 25, 4, 25, 0, 17, 4, 1, 16, 9, 4, 9)  \\[3pt]		
			$(3,1)$   &  	(0, 25, 4, 1, 16, 1, 4, 25, 0, 9, 4, 17, 16, 17, 4, 9)
			 \\[3pt]
			$(3,3)$    &  (0, 17, 4, 25, 16, 25, 4, 17, 0, 1, 4, 9, 16, 9, 4, 1) \\[3pt]
			$(3,9)$   &  (0, 9, 4, 17, 16, 17, 4, 9, 0, 25, 4, 1, 16, 1, 4, 25) \\[3pt]
			$(3,11)$  &   (0, 1, 4, 9, 16, 9, 4, 1, 0, 17, 4, 25, 16, 25, 4, 17) \\[3pt]
			$(7,1)$   &  (0, 17, 4, 9, 16, 25, 4, 1, 0, 1, 4, 25, 16, 9, 4, 17)\\[3pt]
			$(7,3)$ &  (0, 25, 4, 17, 16, 1, 4, 9, 0, 9, 4, 1, 16, 17, 4, 25)  \\[3pt]
				$(7,5)$  &   (0, 9, 4, 1, 16, 17, 4, 25, 0, 25, 4, 17, 16, 1, 4, 9) \\[3pt]
			 $(7,7)$  &   (0, 1, 4, 25, 16, 9, 4, 17, 0, 17, 4, 9, 16, 25, 4, 1) \\[3pt]
				\hline	
			\end{tabular}
		\end{center}
	\end{table}

(ii)
	When $N=18=3^2 \times 2$, that is, $p=3$, $n=2$ and  $q=2$, it follows that PPs $\pi(x) = x^3+ax$ over $\mathbb{Z}_N$   with $a \in \{4,10,16\}$, and the coefficients of QPPs $f_2x^2+f_1x$  over $\mathbb{Z}_N$  satisfy  $f_1\in\{2,4,8,10,14,16\}$ and $f_2\in\{3,9,15\}$, or $f_1\in\{1,5,7,11,13,17\}$ and $f_2\in\{6,12\}$.
	For $\pi_1=x^3+4x \,(\mathrm{mod}\,\, 18)$, $\pi_2=x^3+10x \,(\mathrm{mod}\,\, 18)$ and $\pi_3=x^3+16x \,(\mathrm{mod}\,\, 18)$, let $U_{18}=W_{18}^u$ with	$u\in\{1, 5, 7, 11, 13, 17\}$.  When $u=1$,
the interleaved  ZC sequences are $\s\circ \pi_i =(\xi_{2N}^{e^{(i)}_0}, \xi_{2N}^{e^{(i)}_1}, \dots, \xi_{2N}^{e^{(i)}_{N-1}} ) , i=1,2,3$,
  where
  the exponent sequences
$$
\begin{array}{lll}
\textbf{e}^{(1)} = (e_0^{(1)},e_1^{(1)},\dots, e_{N-1}^{(1)})
= (0, 25, 4, 9, 28, 1, 0, 13, 16, 9, 16, 13, 0, 1, 28, 9, 4, 25    ),\\
\textbf{e}^{(2)} = (e_0^{(2)},e_1^{(2)},\dots, e_{N-1}^{(2)})
= (0, 13, 28, 9, 16, 25, 0,
1, 4, 9, 4, 1, 0,
25,  16, 9, 28, 13        ),\\
\textbf{e}^{(3)} = (e_0^{(3)},e_1^{(3)},\dots, e_{N-1}^{(3)})
= (0, 1, 16, 9, 4, 13, 0,
 25, 28, 9, 28, 25,  0,
  13, 4, 9, 16, 1        ).
  \end{array}$$
 Then we can verify that
	$\s\circ \pi_1$,  $\s\circ \pi_2$ and $\s\circ \pi_3$	
are inequivalent to QPPs and their inverses interleaved ZC
sequences.

(iii)
Take the value of $N$ as $16=2^4$, $27=3^3$ and $150=5^2\cdot 3 \cdot 2$.
 Table \ref{tab_ex} presents the PPs $\pi$ and the proportion of the interleaved ZC sequences  in Theorems \ref{mainresult} and \ref{mainresult2} that are
 inequivalent to those from QPPs and their inverses.

\begin{table}[!ht]
	\small
	\caption{CAZAC sequences $\s\circ \pi$ inequivalent to those from QPPs and their inverses.}\label{tab_ex}
	\begin{center}
		\begin{tabular}{|c|c|c| p{0.5cm}|}
			\hline
			$N$ & $\pi \text{ over } \mathbb{Z}_N$ &     Percent of   \\ %&
 &  &     inequivalent seq.s   \\ %& $\mathcal{N}$\\
			\hline
			\hline
					$16=2^4$ & PPs in Table \ref{tab_3}   & 50 \%  \\
		 &   & \\
		 	 \hline
		 				$27=3^3$ 	  &  $x^3+ax +b,$  & 100 \% \\ [4pt]
		 & $a=1,4,7,10,13,16,19,22,25.$ &  \\[5pt]
			\hline
			               &  $x^5+ax+b,$   &     \\ [4pt]
$150$ &   $a=6,12,  16,  18,  22, 28,   36,42,   46,  48,$  & 100 \% \\ [4pt]	
				$=5^2\times 3\times 2$	      & $	52, 58,66,72,76,   78,  82, 88, 96,102,$  &   \\ [5pt]
					      & $	106, 108, 112,118,126,132,136,138,142,148.$  &    \\ [5pt]	
			\hline
			%	$35=5\times 7$ & $k^7+ak\,(\mathrm{mod}\,\, 35), a\,\mathrm{mod}\,\, 7\neq -1$ & ZCZ, the length of  &  6  \\
			%		 &   &   zero-correlation zone 7 &    \\
			%		 \hline
%			$35=5\times 7$ & $\pi_1(k)=k^7+10k$ &  Unique ZCZ sequences & 1/3  \\
%			 & $\pi_2(k)=k^7+15k$ & in Remark \ref{unique_ZCZ}   &  \\
%			\hline
			%\hline
		\end{tabular}
	\small{
\begin{itemize}
\item  Percent=$\frac{ \# \text{inequivalent CAZAC sequences} }{ \#\text{total CAZAC sequences} }$,
where $\#A$ means the number of set $A$.
\end{itemize}
	}
	\end{center}
\end{table}		
\end{example}

\begin{remark}
Interleaved ZC sequences from different QPPs in \cite{Berggren_2024_} can be shift equivalent.
In this paper, when $N=p^nN_1$ with $n\geq 2$,  distinct
PPs  $\pi_1=x^p+ax$ and $\pi_2=x^p+a'x$ over $\mathbb{Z}_N$  with $a \neq a'$ can generate shift-inequivalent CAZAC sequences, namely,
CAZAC sequences $\s\circ \pi_1$ are shift inequivalent to $\s\circ \pi_2$;
and $\s\circ \pi_1^{-1}$ are shift inequivalent to $\s\circ \pi_2^{-1}$.
%However, no analogous property holds  when $N=2^n$.
\end{remark}

\section{Aperiodic auto-correlation of interleaved ZC sequences}\label{Sec6666}
There is sustained interest in sequences with low aperiodic correlation to better distinguish signals in practical applications \cite{aperiodic_LM}.
Nevertheless, this topic involves the   estimate  of partial exponential sums, which is a well-known challenging problem in analytic number theory.
There are only limited results on the topic. Readers can refer to \cite{aperiodic_HV,aperiodic_Schdmit,aperiodic_Katz} and references therein.

The aperiodic auto-correlation of original ZC sequences has been intensively studied, as reviewed in \cite{Mow_ap}, where Mow and Li provided a more accurate asymptotic bound
on its maximum magnitude as follows:
$$ \lim \limits_{N \rightarrow \infty}  \dfrac{ \max\limits_{0<d<N}|\tilde{\theta}(d)|}{\sqrt{N}}< 0.480261. $$
For the decimated ZC sequences with the root index  $u\neq 1$, the maximum auto-correlation magnitude satisfies
 \begin{equation*}%\label{bound_U}
  \max\limits_{0<d<N}|\tilde{\theta}(d)| \lesssim \frac{\sin \sigma_0 }{ \sqrt{ \pi \sigma_0 }} \sqrt{\frac{|b|}{u}} N,
 \end{equation*}
 where $b$ is an integer defined by $bN \equiv -1 \pmod u$ and $ 1 \leq |b| \leq u/2 $, and
 $\sigma_0\approx 1.16556 $ is the first positive root of the equation $\tan \sigma =2\sigma $ \cite{Mow_ap}.
Numerical results in \cite{Antweiler} suggest that ZC sequences possess the lowest maximum aperiodic auto-correlation at $u=1$.
%This demonstrates that the choice of root  indices
%impacts the aperiodic auto-correlation of ZC sequences.
%From \cite{Mow_ap}, we know that the aperiodic auto-correlation  property of the original ZC sequences is good.

The root indexes  not only impact the aperiodic autocorrelation of ZC sequences, but also
 Frank sequences and Golomb sequences.
Frank sequences and Golomb sequences
 exhibit  good aperiodic autocorrelation  $k\sqrt{N}$   at the root index  $u=1$ \cite{Fan_Frank, Gabidulin}, where $k<1$.
However,
for Frank sequences of odd length $N$  at the root index $u=2$,
 their   aperiodic autocorrelation  \cite{Fan_Frank}   satisfies
$\lim \limits_{N \rightarrow \infty}\max\limits_{0<d<N} |\tilde{\theta}(d)|=\frac{2N}{\pi^2}$.
 Similarly,  for Golomb sequences with the root index $u \geq 2$ and $0.37 \leq b/u \leq 0.5 $ with $bN \equiv \pm 1 \pmod u$,
 their  aperiodic autocorrelation  satisfies
 $\lim \limits_{N \rightarrow \infty} \max\limits_{0<d<N}|\tilde{\theta}(d)| \simeq  \frac{N}{\pi} \sin \left(\frac{\pi b}{u} \right)$  given in
 \cite[Equality (10)]{Gabidulin}.

Below we shall investigate
the aperiodic auto-correlation of a class of QPP-interleaved ZC sequences
% $\s \circ \pi$  as in Theorem \ref{mainresult2}
for any root index $u$.
We first give some important lemmas, which will be used in the estimation of the aperiodic auto-correlation.

\begin{lem}(\!\!\cite[Lem.~1]{Korobov_book})\label{Weyl_sum}
Let $\alpha$ be an arbitrary real number, $Q$ an integer, and $P$ a positive integer.
Then
$$
\begin{array}{c}
\left | \sum\limits_{x=Q+1}^{Q+P}e^{2 \pi i \alpha x } \right| \leq \min(P, \frac{1}{2\| \alpha\|}),
\end{array}
$$
where $\| \alpha\|$ is the distance from $\alpha$ to the nearest integer, i.e.,
$\| \alpha\| = \min\limits_{n \in \mathbb{Z}} |\alpha -n|$.
\end{lem}

\begin{lem}(\!\!\cite[Lem.~13]{Korobov_book})\label{Product_solutions}
Let $\lambda$ and $x_1, x_2,  \dots , x_m$ be positive integers. Denote by $\Gamma_m(\lambda)$ the number
of solutions of the equation $x_1x_2 \cdots x_m=\lambda$. Then for any $\varepsilon \,(0 < \varepsilon  \leq 1)$ we have
$$ \Gamma_m(\lambda)\leq C_m(\varepsilon)\lambda^\varepsilon
$$
where the constant $C_m(\varepsilon)= (\frac{m}{\varepsilon \ln{2}})^{m e^{m/\varepsilon}}$ depends on $m$ and $\varepsilon$ only.
\end{lem}

\begin{lem}\label{one_parsum}
Let $P\geq 3$ and $\alpha={a}/{q}$ with $\gcd(a, q)=1$.
Let $k$ be a positive integer such that $\gcd(k,q)=g\leq 4$.
Then for any positive integer $Q$, real numbers $\beta$ and $\varepsilon \, (0<\varepsilon \leq 1)$, we have
$$
\begin{array}{c}
\sum\limits_{x=0}^{Q-1}  \min(P, \frac{1}{2\| k \alpha x +\beta \|})
\leq
( 1+  \frac{gQ}{q}) (P+q)\frac{2P^\varepsilon}{\varepsilon}.
\end{array}
$$
\end{lem}

The proof of Lemma \ref{one_parsum} is given in
\hyperref[appendix:G]{Appendix G}.
Based on the above preparations, we present  the magnitude of  aperiodic auto-correlation $|\tilde{\theta}(d)|$ of $\s \circ \pi$, where $\s \circ \pi$ are QPP-interleaved ZC sequences in
Theorem \ref{mainresult2} (i).

%When $N=2^n$,  interleaved sequences  $\s \circ \pi$ are  ZAC sequences showed in
%Theorem \ref{mainresult2} (i).
%Next, we investigate  the   aperiodic auto-correlation magnitude  $|\tilde{\theta}(d)|$ of $\s \circ \pi$  with  QPPs  $\pi$.
\begin{thm}\label{aperiod}
Let  $N=2^n$ with $n\geq 2$.
Let
$\s=(s(k))_{0\leq k <N}$ with $s(k)=\xi_N^{uk^2/2}$ be a ZC sequence of length $N$
 and  $\pi(x) =  a_2x^2+a_1x+a_0 $ be a QPP over $\mathbb{Z}_N$,
 where $a_2$ is even, and $a_1, u$ are odd.
% where  $2\, | \, a_2$, $\gcd(a_1,2)=1$ and $\gcd(u,2)=1$.
Then for any shift $d$,
 the  aperiodic auto-correlation magnitude $|\tilde{\theta}(d)|$
of the interleaved sequence  $\s \circ \pi$ satisfies
$$|  \tilde{\theta}(d) 	 |
\leq C(\varepsilon, \varepsilon_1) N^{1-\frac{\varepsilon_1-\varepsilon }{4}},
$$
where $C(\varepsilon,\varepsilon_1)
= 2^{\frac{\varepsilon_1-\varepsilon }{4}} \sqrt[4]{2 + \frac{24}{\varepsilon 4^{\varepsilon/3}}  (\frac{6}{\varepsilon \ln{2}})^{2 e^{6/\varepsilon}} }
$ with $0 <\varepsilon < \varepsilon_1 \leq \frac{\varepsilon +1}{2} <1$.
\end{thm}
\begin{proof}
The derivation will proceed in three steps. First,
an observation that the involved polynomials in the correlation function $|\tilde{\theta}(d)|$ are PPs
 is utilized to simplify the exponential sum when it applies. Next,
 following Weyl's method~\cite{Korobov_book}, we
investigate the second-order finite difference of the involved polynomial $f(x)$, which transforms the modulus of the aperiodic correlation to the Weyl sums of linear polynomials.
 Finally, by applying Lemmas \ref{Weyl_sum}, \ref{Product_solutions} and \ref{one_parsum},
several sums in the investigation of $|\tilde{\theta}(d)|$ will be upper bounded by a single sum, which is further simplified to
yield the derived upper bound.

According to the definition, we have	
\begin{align}\label{key}			
|\tilde{\theta}(d)| 	&=		|\sum\limits_{k=0}^{N-d-1} s(\pi(k))s^{*}(\pi(k+d))|     \nonumber \\
	&=		|\sum\limits_{k=0}^{N-d-1} s(\pi(k+d))s^{*}(\pi(k))  |     \nonumber \\
&	=     | \sum\limits_{k=0}^{N-d-1} \xi_N^ { u(  \pi(k+d)^{2} -\pi(k)^{2})/2} |  \nonumber \\
&	= 	| \xi_N^ { u(\frac{a_2^2}{2}d^{4} +  a_2a_1d^{3}+ \frac{a_1^2}{2} d^2 ) }  \sum\limits_{k=0}^{N-d-1} \xi_N^{ u[2a_2^2dk^3 + 3a_2d(a_2d+a_1)k^2 + (2a_2^2d^2+a_1^2+3a_2a_1d)dk]}|  \nonumber \\
%&	=  |\sum\limits_{k=0}^{N-d-1}  \xi_N^{u[2a_2^2dk^3 + 3a_2d(a_2d+a_1)k^2 + (2a_2^2d^2+a_1^2+3a_2a_1d)dk] }|  \nonumber \\
&=  |\sum\limits_{k=0}^{N-d-1}  \xi_{N}^{ H_d(k) }|			
\end{align}
where
$$H_d(k)=ud[2a_2^2k^3 + 3a_2(a_2d+a_1)k^2 + (2a_2^2d^2+a_1^2+3a_2a_1d)k].$$
	For the exponent polynomial $\frac{H_d(k)}{\gcd(d,N)}=b_3k^3 + b_2k^2 + b_1k$, since $a_2$ is even, $a_1$ and $u$ are odd,
	its coefficients $b_i, i=1,2,3$ satisfy
$$
\left\{
	\begin{array}{cll}
		b_3  &=  2ua_2^2 \frac{d}{\gcd(d,N)}   \equiv 0  \pmod{2} ,\\
		b_2 &=3ua_2(a_2d+a_1) \frac{d}{\gcd(d,N)}  \equiv 0  \pmod{2},    \\
		b_1  &=u(2a_2^2d^2+a_1^2+3a_2a_1d)\frac{d}{\gcd(d,N)}  \equiv ua_1^2\frac{d}{\gcd(d,N)}  \pmod{2} \equiv 1 \pmod{2}.	
	\end{array}\right.
$$	
Then it follows from Lemma \ref{N=2^n}  that	the exponent polynomial $b_3k^3 + b_2k^2 + b_1k$ permutes $\mathbb{Z}_{N_1}$, where $N_1=\frac{N}{\gcd(d,N)}$ is a power of 2.
	It implies
\begin{equation}\label{PP_zero}
\sum\limits_{k=0}^{N_1-1}  \xi_{N_1}^{ b_3k^3 + b_2k^2 + b_1k }=0.
\end{equation}
Moreover let 	$N_2= N-d \, \pmod{N_1} $,
	then \eqref{key} is reduced to
$	|  \tilde{\theta}(d) 	 |=\, |\sum\limits_{k=0}^{N_2-1}  \xi_{N_1}^{ b_3k^3 + b_2k^2 + b_1k }| .$
When $N_2 \leq  N_1/2$, the number of terms in the sum is less than or equal to $N_1/2$.
When $N_2 > N_1/2$, from \eqref{PP_zero}  we have
\begin{align*}
		\left|  \tilde{\theta}(d) 	 \right|
%		=\, |\sum\limits_{k=0}^{N_2-1}  \xi_{N_1}^{ b_3k^3 + b_2k^2 + b_1k  }|
=\, \left|\sum\limits_{k=0}^{N_1-1}  \xi_{N_1}^{ b_3k^3 + b_2k^2 + b_1k  } - \sum\limits_{k=N_2}^{N_1-1}  \xi_{N_1}^{ b_3k^3 + b_2k^2 + b_1k  }  \right|
=\, \left| \sum\limits_{k=N_2}^{N_1-1}  \xi_{N_1}^{ b_3k^3 + b_2k^2 + b_1k  }  \right|.
	\end{align*}
Therefore, we can unify the two cases as
\begin{align*}
		\left|  \tilde{\theta}(d) 	 \right|
		=\, \left|\sum\limits_{k=0}^{N_2-1}  \xi_{N_1}^{ b_3k^3 + b_2k^2 + b_1k  }\right|
		=\, \left|\sum\limits_{k=m_1}^{m_2-1}  \xi_{N_1}^{ b_3k^3 + b_2k^2 + b_1k   }\right|
=\, \left|\sum\limits_{k=m_1}^{m_2-1}  \xi_{N_1}^{H_d(k)/ \gcd(d,N) }\right|
=\, \left|\sum\limits_{k=m_1}^{m_2-1}  \xi_N^{H_d(k)}\right|,
	\end{align*}
where $m_1=0,   m_2=N_2 \leq N_1/2$ or $m_1=N_2 \geq N_1/2,  m_2= N_1$ satisfy $m_2-m_1 \leq  N_1/2.$

Let
\begin{equation}\label{eqrrrrr}
P=m_2-m_1 \leq N_1/2 =\frac{N}{2\gcd(d,N)}.
\end{equation}
Then
%$P=N-d-1$ or $N-d \, \pmod{ \frac{N}{\gcd(d,N)} }-1 \leq N_1-2 \leq N-2$\\
\begin{equation}\label{result_d}
|  \tilde{\theta}(d) 	 |=\, |\sum\limits_{k=m_1}^{m_2-1}  \xi_N^{H_d(k)}|
=\, |\sum\limits_{k=0}^{m_2-m_1-1}  \xi_N^{H_d(k+m_1)}|
=|\sum\limits_{x=0}^{P-1}  e^{2 \pi i f(x) }|,
\end{equation}
where
$$f(x) = \frac{H_d(x+m_1)}{N} =\alpha_3 x^3 + \alpha_2 x^2 + \alpha_1 x +c$$
with a constant $c$ and  %$\alpha_1 = ud( 6a_2^2m_1^2 + 6a_2m_1(a_2d+a_1)+2a_2^2d+a_1^2+3a_2a_1d   ),  $
 $   % \alpha_2 = \frac{3uda_2(2a_2m_1+a_2d+a_1)/\gcd(N,da_2  )  }{N /\gcd(N,da_2  )  }\text{ and }
      \alpha_3 =\frac{2uda_2^2 /\gcd(N,2da_2^2  )  }{N /\gcd(N,2da_2^2  )  }.$
Write $\alpha_3$ as
\begin{equation*}%\label{eq_coef_3}
\alpha_3 =\frac{a}{q},   \,\, \gcd(a, q)=1.
\end{equation*}
It is clear that
$ q=  \frac{N}{2\gcd(N,da_2^2  )} \leq  \frac{N}{8}
$ since $a_2$ is even.

In the following, we divide the discussion for \eqref{result_d} into three cases
according to the relation of $q$ and $P$ for an arbitrarily chosen $\varepsilon_1$ with $0<\varepsilon_1<1$.

{\bf Case 1: }$q > P^{3-\varepsilon_1} (0< \varepsilon_1 <1) $.
%Then it follows $P < q^{\frac{1}{3-\varepsilon_1}}  $.
By \eqref{result_d}, we easily derive
$$|  \tilde{\theta}(d) 	 | \leq P <q^{\frac{1}{3-\varepsilon_1}}  \leq (\frac{N}{8})^{\frac{1}{3-\varepsilon_1}}\leq \frac{1}{2 \sqrt{2}} N^{\frac{1}{3-\varepsilon_1}}.
$$

{\bf Case 2: }$P^{\varepsilon_1} \leq q \leq P^{3-\varepsilon_1} (0< \varepsilon_1 <1) $.
Recall the second-order finite difference  of $f(x)$
 defined  at the beginning of Section \ref{SEC_2.1}.
It follows that
 $$\Delta_{y_1,y_{2}}   f(x)=3! \alpha_3 y_1 y_{2}x +\beta =6 \alpha_3 y_1 y_{2}x +\beta, $$
where $\beta $ depends only on coefficients of the polynomial $f(x)$ and on quantities $y_1, y_{2}$. Following \cite[Lem.~12]{Korobov_book}, we have
\begin{equation}\label{first_esm}
\left|\sum\limits_{x=0}^{P-1 }  e^{2 \pi i f(x) }\right|^4 \leq 16 P \sum\limits_{y_1=0}^{P_1-1}\sum\limits_{y_2=0}^{P_2-1}
\left|\sum\limits_{x=0}^{P_3-1} e^{2 \pi i  \Delta_{y_1,y_{2}}   f(x) }  \right|,
\end{equation}
where $P_1=P$, $P_2=P_1-y_1$, $P_3=P_1-y_1-y_2$.
According to Lemma \ref{Weyl_sum}, we have
$$\left|\sum\limits_{x=0}^{P_3-1} e^{2 \pi i  \Delta_{y_1,y_{2}}   f(x) }  \right|
=\left|\sum\limits_{x=0}^{P_3-1} e^{2 \pi i  6 \alpha_3 y_1 y_{2}x  }  \right|
\leq  \min(P_3, \frac{1}{2\| 6 \alpha_3 y_1 y_{2}\|}).
$$
Substituting this estimate into \eqref{first_esm} and  singling out the summands
for which $y_1=0$ or $y_2=0$,
%, in which quantities $y_1, y_2$ being equal to zero occur,
we get
\begin{align}\label{second_esm}
\frac{1}{16 P} \left|\sum\limits_{x=0}^{P-1 }  e^{2 \pi i f(x) }\right|^4
& \leq
\sum\limits_{y_1=0}^{P_1-1}P_3
+\sum\limits_{y_2=0}^{P_2-1}P_3
-P_3
+\sum\limits_{y_1=1}^{P_1-1} \sum\limits_{y_2=1}^{P_2-1}
\left|\sum\limits_{x=0}^{P_3-1} e^{2 \pi i  \Delta_{y_1,y_{2}}   f(x) }  \right|  \nonumber \\
& \leq (P_1+P_2-1)P_3
+\sum\limits_{y_1=1}^{P_1-1} \sum\limits_{y_2=1}^{P_2-1}
\min(P_3, \frac{1}{2\| 6 \alpha_3 y_1 y_{2}\|}).
\end{align}

For the first item in \eqref{second_esm}, since $0\leq  y_1$, $0\leq  y_2 \leq P_2 -1=P-y_1-1$, it follows that
$0\leq y_1+ y_2 \leq P-1$, and then
\begin{align}\label{before_esti}
  (P_1+P_2-1)P_3 &= (2P-1-y_1)(P-y_1-y_2)  \leq 2P^2.
%   &=(2P-1)(P-(y_1+y_2))+y_1((y_1+y_2)-P)  \nonumber \\
% &\leq (2P-1)P -y_1
\end{align}

For the second item in \eqref{second_esm},  collect summands with the same values of the product $y_1,y_2$.
By Lemma \ref{Product_solutions}, for the number of solutions of the equation
$y_1y_2=\lambda$, one has the estimate $ \Gamma_2(\lambda)\leq C_2(\varepsilon)\lambda^\varepsilon
$ with $C_2(\varepsilon)= (\frac{2}{\varepsilon \ln{2}})^{2 e^{2/\varepsilon}}$ under any positive $\varepsilon \leq 1$. Then
\begin{align*}%\label{thrid_esm}
\sum\limits_{y_1=1}^{P_1-1} \sum\limits_{y_2=1}^{P_2-1} \min(P_3, \frac{1}{2\| 6 \alpha_3 y_1 y_{2}\|}) & \leq
\sum\limits_{\lambda=y_1y_2=1}^{P^2/4} \Gamma_2(\lambda) \min(P, \frac{1}{2\| 6 \alpha_3 \lambda \|})    \nonumber \\
& \leq
C_2(\varepsilon)(P^2/4)^\varepsilon  \sum\limits_{\lambda=1}^{P^2/4}  \min(P, \frac{1}{2\| 6 \alpha_3 \lambda \|}).
\end{align*}
%for $\alpha_3$ given in \eqref{eq_coef_3},
Moreover,  according to Lemma \ref{one_parsum},
%one has
%$$\sum\limits_{\lambda=1}^{P^2/4} \min(P, \frac{1}{2\| 6 \alpha_3 \lambda \|})\leq ( 1+  \frac{P^2}{2q}) (P+q)\frac{2P^\varepsilon}{\varepsilon}.$$
%Then,
it follows that
%combining \eqref{thrid_esm}, \eqref{fourth_esm} and \eqref{fifth_esm} yields
\begin{equation*}%\label{after_esti00}
\begin{array}{c}
\sum\limits_{y_1=1}^{P_1-1} \sum\limits_{y_2=1}^{P_2-1} \min(P_3, \frac{1}{2\| 6 \alpha_3 y_1 y_{2}\|})
 \leq
 C_2(\varepsilon) \frac{P^{2\varepsilon}}{4^{\varepsilon}}  ( 1+  \frac{P^2}{2q}) (P+q)\frac{2P^\varepsilon}{\varepsilon}
 \leq \frac{8C_2(\varepsilon)}{\varepsilon 4^{\varepsilon}}  P^{3-\varepsilon_1+3\varepsilon},
 \end{array}
\end{equation*}
where
$$
\begin{array}{c}
( 1+  \frac{P^2}{2q}) (P+q)=
P+q+  \frac{P^3}{2q} + \frac{P^2}{2}
\leq   P+P^{3-\varepsilon_1}+   \frac{P^3}{2P^{\varepsilon_1}} +      \frac{P^2}{2}
\leq 4P^{3-\varepsilon_1},
\end{array}
$$
by the condition $P^{\varepsilon_1} \leq q \leq P^{3-\varepsilon_1} (0< \varepsilon_1 <1) $.
Since the above inequality holds for any $\varepsilon$ with $0<\varepsilon<1$, by further replacing $\varepsilon$ with $\varepsilon/3$, we get
\begin{equation}\label{after_esti}
\begin{array}{c}
\sum\limits_{y_1=1}^{P_1-1} \sum\limits_{y_2=1}^{P_2-1} \min(P_3, \frac{1}{2\| 6 \alpha_3 y_1 y_{2}\|})
 \leq \frac{24 C_2(\varepsilon/3)}{\varepsilon 4^{\varepsilon/3}}  P^{3-\varepsilon_1+\varepsilon}.
 \end{array}
\end{equation}

Therefore,  when $P^{\varepsilon_1} \leq q \leq P^{3-\varepsilon_1}$ with $0 <\varepsilon < \varepsilon_1 <1$,
according to \eqref{second_esm},
by combining \eqref{before_esti} and \eqref{after_esti},
we derive
{\small \begin{align*}
 \left|\sum_{x=0}^{P-1} e^{2\pi i f(x)}\right|^4
  & \leq
16P \left[(P_1+P_2-1)P_3 + \frac{24 C_2(\varepsilon/3)}{\varepsilon 4^{\varepsilon/3}}  P^{3-\varepsilon_1+\varepsilon} \right] %\\
 %  &  \leq
%16P (2 + \frac{24 C_2(\varepsilon/3)}{\varepsilon 4^{\varepsilon/3}}  ) P^{3-\varepsilon_1+\varepsilon }
%\\
&\leq 32 (1 + \frac{12 C_2(\varepsilon/3)}{\varepsilon 4^{\varepsilon/3}} ) P^{4-\varepsilon_1+\varepsilon },
\end{align*}}
implying  that
$$
\begin{array}{c}
 \left|\sum_{x=0}^{P-1} e^{2\pi i f(x)}\right|
\leq
2  \sqrt[4]{2(1 + \frac{12C_2(\varepsilon/3)}{\varepsilon 4^{\varepsilon/3}} )} P^{\frac{4-\varepsilon_1+\varepsilon }{4}}.
\end{array}
%=C(n,\varepsilon, \varepsilon_1)P^{1-\frac{\varepsilon_1-\varepsilon }{4}},
$$
%=2  \sqrt[4]{2+ 32\frac{2^{ e^{2/\varepsilon}}}{ 4^{\varepsilon}
%(\ln{2})^{2 e^{2/\varepsilon}}
% \varepsilon^{1+2 e^{2/\varepsilon}} } }$.

Therefore, from \eqref{eqrrrrr} and  \eqref{result_d},
it follows that
\begin{equation*}%\label{3degree_resu}
\begin{array}{c}
|  \tilde{\theta}(d) 	 |
\leq 2  \sqrt[4]{2(1 + \frac{12C_2(\varepsilon/3)}{\varepsilon 4^{\varepsilon/3}} )}  \left(\frac{N}{2\gcd(d,N)} \right)^{1-\frac{\varepsilon_1-\varepsilon }{4}}
\leq C(\varepsilon, \varepsilon_1) N^{1-\frac{\varepsilon_1-\varepsilon }{4}},
\end{array}\end{equation*}
where $C(\varepsilon,\varepsilon_1)
= 2^{\frac{\varepsilon_1-\varepsilon }{4}} \sqrt[4]{2 + \frac{24}{\varepsilon 4^{\varepsilon/3}}  (\frac{6}{\varepsilon \ln{2}})^{2 e^{6/\varepsilon}} }
$.

{\bf Case 3: } $q=1$. In this case,  $\alpha_3$ is an integer.
According to \eqref{result_d},
we have
\begin{equation}\label{res_dd}
|  \tilde{\theta}(d) 	 |=\left|\sum\limits_{x=0}^{P-1 }  e^{2 \pi i f(x) } \right|  =\left|\sum\limits_{x=0}^{P-1 }  e^{2 \pi i (\alpha_2 x^2+ \alpha_1 x) } \right|,
\end{equation}
with
$$ \alpha_2 =\frac{3uda_2(a_2d+a_1) /\gcd(N,da_2  )  }{N /\gcd(N,da_2  )  } = \frac{a'}{q'}
\text{ and } \gcd(a',q')=1.
 $$
 Since  $ q=  \frac{N}{2\gcd(N,da_2^2  )}=1$,
 we have $q'=N /\gcd(N,da_2 )=2 gcd(N, a_2)  \geq 4$ by even $a_2$.
Together with $q'=N /\gcd(N,da_2 )  \leq N/2 $,
it follows  $4 \leq q' \leq N/2 $.

{\bf Subcase 3.1: }$q' > P^{2-\varepsilon_1} (0< \varepsilon_1 <1) $.
 By \eqref{res_dd}, we easily derive
$$|  \tilde{\theta}(d) 	 | \leq P <q^{\frac{1}{2-\varepsilon_1}}  \leq (\frac{N}{2})^{\frac{1}{2-\varepsilon_1}}\leq \frac{1}{2 } N^{\frac{1}{2-\varepsilon_1}}.
$$

{\bf Subcase 3.2: }$P^{\varepsilon_1} \leq q' \leq P^{2-\varepsilon_1} (0< \varepsilon_1 <1) $.
Similarly as Case 2, we can derive that
$$|  \tilde{\theta}(d) 	 | \leq  2\sqrt{1+ \frac{3}{\varepsilon}}\, P^{1- \frac{\varepsilon_1-\varepsilon}{2}} \leq  2^{\frac{\varepsilon_1-\varepsilon}{2}} \sqrt{1+ \frac{3}{\varepsilon}}\, N^{1- \frac{\varepsilon_1-\varepsilon}{2}}  .
$$

 Since $0 <\varepsilon < \varepsilon_1 \leq \frac{\varepsilon +1}{2} <1$, one has
$N^{1- \frac{\varepsilon_1-\varepsilon}{2}} \geq N^{\frac{1}{2-\varepsilon_1}}$.
Thus, for Case 3, we have
\begin{equation*}%\label{2degree_resu}
|  \tilde{\theta}(d) 	 |  \leq  2^{\frac{\varepsilon_1-\varepsilon}{2}} \sqrt{1+ \frac{3}{\varepsilon}}\, N^{1- \frac{\varepsilon_1-\varepsilon}{2}}  .
\end{equation*}
%
%{\bf Subcase 3.3:} $q'=1$. In this case we  have
%$$
%|  \tilde{\theta}(d) 	 |=|\sum\limits_{x=0}^{P-1 }  e^{2 \pi i f(x) }|  =|\sum\limits_{x=0}^{P-1 }  e^{2 \pi i (\alpha_1 x) }| \leq \min (P, \frac{1}{\| \alpha_1\|}).
%$$

Combining with three cases, we obtain the desired upper bound on $|\tilde{\theta}(d)|$.
\end{proof}

 %We need some experimental results or remark here.
\begin{remark}
From the experimental data, we note that the theoretical upper bound in Theorem \ref{aperiod} can be further improved, possibly due to the use of the triangle inequality
$| \sum_{i} a_i |  \leq \sum_{i} |a_i|$ in \eqref{first_esm}.
\end{remark}

\section{Conclusion}\label{Sec5}
This paper contributes to the construction and aperiodic auto-correlation analysis of CAZAC sequences.
Our main contributions are threefold:
firstly, we proposed two classes of permutation-interleaved ZC sequences
to generate CAZAC sequences;
secondly, we showed the inequivalence relation among ZC sequences, and interleaved ZC sequences from QPPs, the proposed PPs in this paper and their inverses,
which confirms the sufficiency of the conjecture  in \cite{Berggren_2024_}
 and covers many new sequences that cannot be obtained in the work of \cite{Berggren_2024_};
thirdly, we  evaluated  the aperiodic auto-correlation of the QPP-interleaved ZC sequences of length $N=2^n$.

The permutation-based interleaving technique appears powerful in generating CAZAC sequences.
There are several interesting research problems, including
  the construction of new permutation polynomials over integer rings,
 the characterization of
permutation polynomials that preserves  the ZAC property of permutation-interleaved ZC sequences, and the investigation of the
 equivalence relation between permutation-interleaved ZC sequences and those from the unified construction by Mow \cite{Mow_unif}.

%\section*{Acknowledgments}
%The work of X. Zeng was supported by the Innovation Group Project of the Natural Science Foundation of Hubei Province of China under Grant 2023AFA021, and
%by the National Natural Science Foundation of China under Grant 12471492.

%\bibliography{mustbe}
%\bibliographystyle{ieeetr}	

\section*{Appendix A: Proof of Lemma \ref{N=pq}}\label{appendix:A}
\begin{proof}
	 When $2\,|\,n$, it follows that $\frac{n^i}{2}\equiv 0 \ \;(\mathrm{mod}\,\,n)$  with $i\geq 2$. Then we have
	$$
\begin{array}{cll}
	\frac{1}{2} (m+sn)^{2p}
&= &\frac{1}{2} [ m^{2p}+(sn)^{2p}] + p m^{2p-1}sn + \frac{1}{2} \sum\limits_{ i=2 }^{ 2p-1} \binom{2p}{i} m^{2p-i}(sn)^i   \\
&\equiv &\frac{1}{2} [ m^{2p}+(sn)^{2p}] \ \;(\mathrm{mod}\,\,n).
	\end{array}
$$

	 When $2\nmid n$,
according to Lemma \ref{parity} (ii),
we have $2\,|\,\binom{2p}{i}$ with odd $i$,
and $ \binom{2p}{2j} x^2 n$ and $ \binom{p}{j}x$ have the same parity for $1\leq j\leq p$ and an integer $x$.
Then it implies that		
		\begin{align*}	
			&\frac{1}{2}[(m+sn)^{2p}+(m+sn)^{p}] \\
			= \,\,&\frac{1}{2}(m^{2p}+m^{p})+
			\frac{1}{2}\big[ \sum\limits_{ \text{odd } i}^{1\leq i\leq 2p-1} \binom{2p}{i} m^{2p-i}(sn)^i \big] \\
\,\,& +\frac{1}{2} \big[\sum\limits_{ j=1}^{p} \binom{2p}{2j} m^{2p-2j}(sn)^{2j}      \big]
 +\frac{1}{2} \big[\sum\limits_{ j=1}^{p} \binom{p}{j} m^{p-j}(sn)^j      \big]\\
\equiv\,\, &\frac{1}{2}(m^{2p}+m^{p})	
 + \sum\limits_{ j=1}^{p} \frac{1}{2} \big[\binom{2p}{2j} m^{2p-2j}s^{2j}n^{2j-1} +\binom{p}{j} m^{p-j}s^{j}n^{j-1}   \big] n  \ \; (\mathrm{mod}\,\,n)\\
\equiv\,\, & \frac{1}{2}(m^{2p}+m^{p}) \ \; (\mathrm{mod}\,\,n).
	\end{align*}
\end{proof}

\section*{Appendix B:  Proof of Proposition \ref{corre_thm} }\label{appendix:B}
\begin{proof}
Let $p$ be any prime factor of $N$.
Since $\pi(x)=x^p+ax+b$ is a PP over $\mathbb{Z}_{N}$,
from the proof of Lemma \ref{N=p^n} (i),
 we have $a \not\equiv -1 \,\,(\mathrm{mod}\,\, p)$  when $n =1 $;
  and $a \not\equiv 0,-1  \,\,(\mathrm{mod}\,\, p)$ when $n \geq 2$.
According to Lemma  \ref{N=pq}, we have
$$\frac{1}{2} (m+N)^{2p}+ \frac{1}{2} (N\mathrm{mod}\,\,2 )(m+N)^{p} \equiv \frac{1}{2}m^{2p}+ \frac{1}{2} (N\mathrm{mod}\,\,2 )m^{p} \, \, (\mathrm{mod}\,\,N)$$
  for $m=k$ or $k+d$. Then
$$
	\begin{array}{lll}
	h_d(k+N)
	&=&\frac{1}{2}[(k+N+d)^{2p}-(k+N)^{2p}] +a[(k+N+d)^{p+1}-(k+N)^{p+1}] +\\
	&& \frac{1}{2}(N\mathrm{mod}\,\,2 ) [(k+N+d)^{p}-(k+N)^{p}]+a^2(k+N)d    \\
		&\equiv & \frac{1}{2} [(k+N+d)^{2p}+ (N\mathrm{mod}\,\,2 ) (k+N+d)^{p}]
-\frac{1}{2} [(k+N)^{2p}+     \\
	&& (N\mathrm{mod}\,\,2 ) (k+N)^{p} ]+a[(k+d)^{p+1}-k^{p+1}] +a^2dk  \ \; (\mathrm{mod}\,\,N)  \\
	&\equiv&\frac{1}{2}[(k+d)^{2p}-k^{2p}]  +a[(k+d)^{p+1}-k^{p+1}] \\
&& +\frac{1}{2}(N\mathrm{mod}\,\,2 ) [(k+d)^{p}-k^{p}] +a^2dk \ \; (\mathrm{mod}\,\,N) \\
		&\equiv&h_d(k)\ \; (\mathrm{mod}\,\,N).
	\end{array}
	$$
	Then the claimed statement follows. It implies that $\zeta_N^{h_d(k)} = \zeta_N^{h_d(k+N)}$ for any $k$.
	Thus,	for any integer $t$, we have $\theta(d) =C_0\sum^{N-1}\limits_{k=0} \zeta_N^{h_d(k+t) } $ since the sum is independent of the order
	of its terms.	
	Let $t=\frac{2N}{\mathrm{gcd}(d, N)p}$. Then $t$ is an integer by the condition $p^n \nmid d$. 	

Next, we  investigate  $(h_d(k+t)-h_d(k))\ \mathrm{mod}\,\,N $. From \eqref{h_d(k_form)}, we have
%extract $h_d(k)$ from the exponent $h_d(k+t)$, then consider the rest exponent part.
$$
	\begin{array}{lll}
		h_d(k+t)- h_d(k) =m_d^{(1)}(k)+m_d^{(2)}(k)+m_d^{(3)}(k)+a^2td,
	\end{array}
	$$
	where
	\begin{equation*}%\label{md123}
	\begin{array}{cll}
		m_d^{(1)}(k)&=& \frac{1}{2}\sum\limits_{i=1}^{2p-1} \binom{2p}{i} t^i \big[ (k+d)^{2p-i} -k^{2p-i} \big],\\
		m_d^{(2)}(k)&=& a \sum\limits_{i=1}^{p} \binom{p+1}{i} t^i \big[ (k+d)^{p+1-i}-k^{p+1-i} \big],\\
		m_d^{(3)}(k)&=& \frac{1}{2}(N\mathrm{mod}\,\,2 ) \sum\limits_{i=1}^{p-1} \binom{p}{i} t^i \big[ (k+d)^{p-i} - k^{p-i}\big].
	\end{array}
	\end{equation*}
	Thus
	\begin{equation}\label{ddd}
		\theta(d) =C_0\sum^{N-1}\limits_{k=0} \zeta_N^{h_d(k+t) } = C_0\sum^{N-1}\limits_{k=0} \zeta_N^{h_d(k) } \zeta_N^{m_d^{(1)}(k)+m_d^{(2)}(k)+m_d^{(3)}(k)+a^2td}.
	\end{equation}
	In the following, we shall investigate that for the integer $t=\frac{2N}{\mathrm{gcd}(d, N)p}$ defined above,  $ (m_d^{(1)}(k)+m_d^{(2)}(k)+m_d^{(3)}(k)+a^2td)\ \mathrm{mod}\,\,N$
	becomes a nonzero constant  for any $k, 0 \leq k < N $.

For $m_d^i(k), i=1,2,3$, if certain terms in $m_d^i(k)$ have the factor $tdp$, then these terms can be divided by $N$.  In fact,
{\small
\begin{align*}
		m_d^{(1)}(k)
		= \,\,&  \frac{1}{2}\sum\limits_{i=1}^{2p-1} \binom{2p}{i} t^i \big[ (k+d)^{2p-i} -k^{2p-i}\big] \\
		= \,\,&\frac{1}{2}\bigg[ 2p t \bigg( \binom{2p-1}{1} d k^{2p-2} + \binom{2p-1}{2} d^2  k^{2p-3}+\dots + \binom{2p-1}{2p-2}d^{2p-2} k + d^{2p-1} \bigg) \\
		& + {2p\choose 2} t^2 \bigg( \binom{2p-2}{1} d k^{2p-3} + \binom{2p-2}{2} d^2  k^{2p-4}+\dots  + \binom{2p-2}{2p-3} d^{2p-3} k + d^{2p-2} \bigg)\\
&+\dots+ \binom{2p}{2p-2} t^{2p-2} ( 2 d k + d^{2})+ \binom{2p}{2p-1}t^{2p-1}d \bigg].
	\end{align*}
}Note that in the above equality, each term contains $d$.
	When the degree of $t$ is 1, $N\,|\,\frac{1}{2} \binom{2p}{1} td$ with $t=\frac{2N}{\mathrm{gcd}(d,N)p}$;
	When the degree of $t$ is more than 1, if the degree of $d$ is 1, then
	$N\,|\,\frac{1}{2} \binom{2p}{i} \binom{2p-i}{1} t^i d$  by $\binom{2p}{i} \binom{2p-i}{1}=\binom{2p}{1} \binom{2p-1}{i}, 2 \leq i \leq 2p-1$; if the degree of $d$ is more than 1, then
 we have the following results:

 % \begin{enumerate}
  %\item
  --   when $n\geq 2$, it suffices to consider the common factor $\frac{1}{2}t^2 d^2 $ of these terms.
Note that	$ \frac{1}{2}  t^2 d^2= 2N  \frac{N}{p^2} (\frac{d}{\mathrm{gcd}(d,N)})^2$ is divided by $N$ from $N=p^n \geq p^2$. Thus each term in $m_d^{(1)}(k)$ is divided by $N$, which implies  $m_d^{(1)}(k) \equiv 0  \,\, (\mathrm{mod}\,\,N)$.

%\item
-- when $n=1$, since $ p \mid  \binom{2p}{j}$ for $1 \leq j < 2p,\, j \neq p$ given in  Lemma  \ref{parity} (i),
for $j \neq p$ we have	
$$ \frac{1}{2} \binom{2p}{j} t^2 d^2= 2N \frac{\binom{2p}{j}}{p} \frac{N}{p} (\frac{d}{\mathrm{gcd}(d,N)})^2,$$
 which is divided by $N$.
Thus
$$m_d^{(1)}(k) \equiv \frac{1}{2} \binom{2p}{p} t^p \big[ (k+d)^{p}-k^{p} \big] (\mathrm{mod}\,\,N) \equiv  \frac{1}{2} \binom{2p}{p} t^p d^{p}  (\mathrm{mod}\,\,N)   ,$$
 from $p \mid  \binom{p}{i},  0<i<p $ in  Lemma  \ref{parity} (i).
%\end{enumerate}

 Let $\delta=1$ when $n=1$, and $\delta=0$ when $n\geq 2$. Then
 $m_d^{(1)}(k) \equiv \frac{\delta}{2} \binom{2p}{p} t^p d^{p}   \,\, (\mathrm{mod}\,\,N)$.
 For $m_d^{(2)}(k)$, we have
	\begin{align*}
		&m_d^{(2)}(k)\\
		\vspace{1ex}
		=& a \sum\limits_{i=1}^{p} \binom{p+1}{i} t^i \big[  (k+d)^{p+1-i} -k^{p+1-i}\big]\\
				\vspace{1ex}
		= &a \left[(p+1)t \,\bigg(\sum\limits_{i=1}^{p}\binom{p}{i} d^i k^{p-i} \bigg)
		+\binom{p+1}{2} t^2 \,\bigg(\sum\limits_{i=1}^{p-1}\binom{p-1}{i} d^i k^{p-1-i} \bigg)
		+ \dots 	+(p+1) t^{p}d  \right] \\
		= &a \bigg[(p+1)t \,\bigg(\sum\limits_{i=1}^{p-1}\binom{p}{i} d^i k^{p-i} \bigg)
		+\binom{p+1}{2} t^2 \,\bigg( \sum\limits_{i=1}^{p-1}\binom{p-1}{i} d^i k^{p-1-i} \bigg)
		+ \dots 	\\
		& +\binom{p+1}{p-1} t^{p-1} (2dk+ d^2)  +\binom{p+1}{p} t^{p}d  +\binom{p+1}{1}td^p  \bigg].
	\end{align*}	
	For the first term $a(p+1)t \,\bigg(\sum\limits_{i=1}^{p-1}\binom{p}{i} d^i k^{p-i} \bigg)  $  in  $m_d^{(2)}(k)$,
since $t=\frac{2N}{\mathrm{gcd}(d,N)p}$ and $p\,|\,\binom{p}{i} $ with $1\leq i<p$ in Lemma \ref{parity} (i), 	we obtain $N \,| \,td\binom{p}{i}$. Thus the first term in  $m_d^{(2)}(k)$ is divided by $N$.	
Moreover,	Lemma \ref{parity} (i) shows that $p\,|\,\binom{p+1}{j}$ with $2 \leq j< p$,
	and each term in $m_d^{(2)}(k)$ contains $d$.
	Then for $t=\frac{2N}{\mathrm{gcd}(d,N)p}$ and $2 \leq j< p$,
	we see $N \,|\,\binom{p+1}{j}td$.
	Thus from $N \mid tdp$, we have 	
$$m_d^{(2)}(k) \equiv a \bigg(  \binom{p+1}{p}    t^{p}d +\binom{p+1}{1} td^p \bigg) \,\, (\mathrm{mod}\,\,N) \equiv a(t^{p}d + td^p) \, \,(\mathrm{mod}\,\,N).$$	
Similarly, we obtain that
	$m_d^{(3)}(k) \equiv 0 \, (\mathrm{mod}\,\,N)$.
Thus
$$ m_d^{(1)}(k)+m_d^{(2)}(k)+m_d^{(3)}(k)  +a^2td \equiv
    atd( t^{p-1} + d^{p-1}+a)+\delta \frac{\binom{2p}{p}}{2}t^pd^p\, \,  (\mathrm{mod}\,\,N),
$$
which becomes a constant for any $k, 0 \leq k < N$. Finally, we shall show that
\begin{equation}\label{cha55_N}
atd( t^{p-1} + d^{p-1}+a)+\delta \frac{\binom{2p}{p}}{2}t^pd^p \not\equiv 0    \, \,(\mathrm{mod}\,\,N).
	\end{equation}	

When $n=1$, we have $\gcd(d,p)=1$.
It follows from $t=\frac{2N}{\mathrm{gcd}(d,N)p}$ that $\frac{N}{p} \mid td$ and $\gcd(t,p)=1$.
Thus
$td( at^{p-1} + ad^{p-1}+a^2+  \frac{\binom{2p}{p}}{2}t^{p-1}d^{p-1}) \equiv 0 \,(\mathrm{mod}\,\,N)$	 if and only if
$ at^{p-1} + ad^{p-1}+a^2+  \frac{\binom{2p}{p}}{2}t^{p-1}d^{p-1} \equiv 0 \,(\mathrm{mod}\,\,p)$.
In fact, from Euler's Theorem and $\binom{2p}{p} \equiv \binom{2}{1} \,(\mathrm{mod}\,\,p)$ in  Lemma  \ref{parity} (iii), we have
$ at^{p-1} + ad^{p-1}+a^2+  t^{p-1}d^{p-1}
\equiv (a+1)^2 \,(\mathrm{mod}\,\,p)
$, which is nonzero by $a  \not\equiv -1 \, (\mathrm{mod}\,\,p)$.
Hence, the equation in \eqref{cha55_N} holds.

When $n \geq 2$,	
from $a  \not\equiv 0,-1 \, (\mathrm{mod}\,\,p)$, it implies $p\geq 3$.
Suppose  $\mathrm{gcd}(d,N)= p^j q_1^{j_1}q_2^{j_2}\dots q_k^{j_k}$ with $0\leq j\leq n-1$ and $j_1\in \mathbb{Z}_{n_1+1} ,j_2\in \mathbb{Z}_{n_2+1},\dots,j_k\in \mathbb{Z}_{n_k+1}$,
	then we get $d=p^j q_1^{j_1}q_2^{j_2}\dots q_k^{j_k}m$ with $\mathrm{gcd}(m,N)=1$, and
	$$t=\frac{2N}{\mathrm{gcd}(d,N)p}=2p^{n-1-j} q_1^{n_1-j_1}q_2^{n_2-j_2}\dots q_k^{n_k-j_k}.$$
	It follows that $td=2m\frac{N}{p}$.
	Since $p\geq 3$ and $a  \not\equiv 0,-1 \, (\mathrm{mod}\,\,p)$,
if $N \,| \,atd( t^{p-1} + d^{p-1}+a)    $, then $p \,| \,2ma( t^{p-1} + d^{p-1}+a)$.
Below we investigate $2ma (t^ {p-1}+d^ {p-1}+a) \, \, (\mathrm{mod}\, \, p) $
 according to the relations of $t,d$ and $p$.

{\bf 1)}  For the case of $j=0$, we have $p\,|\,t$ and $\mathrm{gcd}(d,p)=1$, following that $d^{p-1} \equiv 1 \,(\mathrm{mod}\,\,p)$ from Euler's Theorem.
Then
$ 2ma( t^{p-1} + d^{p-1}+a)  \equiv 2ma( a+1) \,\,(\mathrm{mod}\,\,p)$. Since $2$, $m$, $a$ and $a+1$ all are prime to $p$, we have $p \nmid 2ma( a+1)$, implying that $ 2ma( t^{p-1} + d^{p-1}+a)\not\equiv 0 \,\,(\mathrm{mod}\,\,p)$.
Thus \eqref{cha55_N} holds.

{\bf 2)} For the case of $1\leq j\leq n-2$, we have $p\,|\,t$ and $p\,|\,d$. Then
$2ma( t^{p-1} + d^{p-1}+a)  \equiv 2ma^2\,(\mathrm{mod}\,\,p)$.
From $\mathrm{gcd}(2ma,p)=1$, we have $p \nmid 2ma^2$, implying that $ 2ma( t^{p-1} + d^{p-1}+a)\not\equiv 0 \,\,(\mathrm{mod}\,\,p)$. Hence \eqref{cha55_N} holds.

{\bf 3)} For the case of $j= n-1$, we have $p\,|\,d$ and $\mathrm{gcd}(t,p)=1$. Then $t^{p-1} \equiv 1\,\,(\mathrm{mod}\,\,p)$ again by Euler's Theorem.
It follows that $ 2ma( t^{p-1} + d^{p-1}+a)  \equiv 2ma( a+1) \,\,(\mathrm{mod}\,\,p)$.
By the analysis of  1), we can similarly show that equation \eqref{cha55_N} holds.

	 According to  the above cases, for each $d$ with $p^n \nmid d$,
	$ (h_d(k+t)- h_d(k)) \,\,\mathrm{mod}\,\,N$  is a  nonzero constant for any $k$.
%Then from equation \eqref{cha55_N}, it indicates that
Then, the constant
	$$C_1=\zeta_N^{m_d^{(1)}(k)+m_d^{(2)}(k)+m_d^{(3)}(k)+a^2td} \neq 1.$$
	Combing \eqref{corehk} and \eqref{ddd}, it follows $\theta(d)=C_1\theta(d)$, which results in
	$\theta(d)=0, 0<d<N$.	
\end{proof}

\section*{Appendix C:  Proof of Theorem \ref{mainresult2} (i)}\label{appendix:C}	
\begin{proof}
Since the constant term $a_0$ in the PP $\pi$ is a mathematical operator in \eqref{mathe_opera} which preserves the CAZAC property,
it suffices to investigate $\pi(k)=a_mk^m + a_{m-1}k^{m-1} + \dots + a_2k^2 + a_1k$.

We can see that
$$
\theta(d) =\, \sum^{N-1}\limits_{k=0}y(k)y^{*}(k + d)
	=\, \sum^{N-1}\limits_{k=0} \xi_N^{\frac{u}{2}\pi^2(k)} \xi_N^{-\frac{u}{2}\pi^2(k+d) }
	=\, \sum^{N-1}\limits_{k=0} \zeta_N^{\frac{\pi^2(k+d)-\pi^2(k)}{2}   }
$$
where $\zeta_N=\xi_N^{-u}$.
Let $H_d(k)=\frac{\pi^2(k+d)-\pi^2(k)-\pi^2(d) }{2}$,
then
\begin{equation}\label{corred}
\theta(d)=\zeta_N^{ \frac{\pi^2(d) }{2} }\, \sum^{N-1}\limits_{k=0} \zeta_N^{ H_d(k)   }
=\zeta_N^{ \frac{\pi^2(d) }{2} }\, \sum^{N-1}\limits_{k=0} \zeta_{N/ \gcd(d,N)}^{ H_d(k) / \gcd(d,N)  }.
\end{equation}
Below we shall prove that $H_d(k)/ \gcd(d,N)$ is a PP over $\mathbb{Z}_{N/ \gcd(d,N)}$ with $N=2^n$.

According to Lemma \ref{N=p^n} (i), $\frac{H_d(k)}{\gcd(d,N)}$ is a PP over $\mathbb{Z}_{2^n}$ if and only if $\frac{H_d(k)}{\gcd(d,N)}$ is a PP over $\mathbb{Z}_{2}$ and  its derivative $H_d'(k)/ \gcd(d,N) \not\equiv 0  \,\,(\mathrm{mod}\,\,2)$ for any $k \in \mathbb{Z}_{2}$.
In the following, the discussion is divided into two cases $2\,|\, d$
and $\gcd(2,d)=1$.

{\bf Case 1.}  $2\,|\, d$.
{\bf 1)} We first prove that $H_d(k)/ \gcd(d,N)$ is a PP over $\mathbb{Z}_{2}$.
Note that
\begin{align}\label{hdk_oe}
  H_d(k)/ \gcd(d,N) = &\frac{\pi^2(k+d)-\pi^2(k)-\pi^2(d)    }{2\gcd(d,N)}  \notag \\
  &  =\frac{\pi(k+d)+\pi(k) }{2} \cdot \frac{\pi(k+d)-\pi(k) }{\gcd(d,N)}
-\frac{\pi^2(d) }{2\gcd(d,N)}.
\end{align}
For the last term in \eqref{hdk_oe}, we have
$$\frac{\pi^2(d) }{2\gcd(d,N)} = (a_m k^{m-1} + a_{m-1}k^{m-2} + \dots + a_2k + a_1)^2  \frac{d^2 }{2\gcd(d,N)}
\equiv a_1^2  \frac{d }{2} \equiv \frac{d }{2}  \,\,(\mathrm{mod}\,\,2) $$
since $2\,|\, d$, $a_1$ and $\frac{d }{\gcd(d,N)}$ are coprime to $2$.

For
$\frac{\pi(k+d)-\pi(k) }{\gcd(d,N)}$ in \eqref{hdk_oe},
since $2 \nmid a_1$ and $2 \nmid  \sum\limits_{\text{ odd } i=1}^{m}  a_i $,
it follows that
\begin{equation}\label{interpro}
 (a_1 + 2a_2k  + 3a_{3}k^{2} + \dots + m a_m k^{m-1}  ) \equiv 1   \,\,(\mathrm{mod}\,\,2)
\text{ for }  k\in \mathbb{Z}_2.
\end{equation}
Due to $\frac{d^2 }{\gcd(d,N)}  \equiv  0 \,\,(\mathrm{mod}\,\,2)$, for any $k\in \mathbb{Z}_2 $, we have
\begin{align}
\frac{\pi(k+d)-\pi(k) }{\gcd(d,N)}
= &  \frac{\sum_{i=1}^{m}   a_i(k+d)^i-a_i k^i   }{\gcd(d,N)} \notag \\
\equiv &   \sum_{i=1}^{m} a_i \binom{i}{1}  k^ {i-1} \frac{ d }{\gcd(d,N)} \,\,(\mathrm{mod}\,\,2)   \notag \\
\equiv   &  (a_1 + 2a_2k  + 3a_{3}k^{2} + \dots + m a_m k^{m-1}  ) \,\,(\mathrm{mod}\,\,2) \notag \\
\equiv & 1   \,\,(\mathrm{mod}\,\,2). \notag
\end{align}

For
$\frac{\pi(k+d) +\pi(k) }{2}$ in \eqref{hdk_oe},
due to $\frac{d^2 }{2}  \equiv  0 \,\,(\mathrm{mod}\,\,2)$, we obtain
\begin{align}
\frac{\pi(k+d) +\pi(k)}{2}
%= &  \frac{\sum_{i=1}^{m}  a_ix^i +  a_i(x+d)^i  }{2} \notag \\
\equiv &   \sum_{i=1}^{m} a_ik^i + \sum_{i=1}^{m} a_i  \binom{i}{1}  k^{i-1} \frac{ d }{2}  \,\,(\mathrm{mod}\,\,2)   \notag \\
\equiv   & \sum_{i=1}^{m} a_ik^i + \frac{ d }{2}   (a_1 + 2a_2k  + 3a_{3}k^{2} + \dots + m a_m k^{m-1}  ) \,\,(\mathrm{mod}\,\,2)      \notag \\
\equiv   & \sum_{i=1}^{m} a_ik^i + \frac{ d }{2} \,\,(\mathrm{mod}\,\,2), \notag
\end{align}where the last step holds by \eqref{interpro}.
Hence,
$\frac{\pi(k+d) +\pi(k)}{2} \,\,(\mathrm{mod}\,\,2) \equiv
\left\{
  \begin{array}{ll}
    \frac{ d }{2}, & {k=0,} \\
    \frac{ d }{2}+1, & {k=1.}
  \end{array}
\right.
$

Combining the above three parts in \eqref{hdk_oe}, for \eqref{hdk_oe} we can derive that
  $$H_d(k)/ \gcd(d,N) \,\,(\mathrm{mod}\,\,2) \equiv  \left\{
  \begin{array}{ll}
    0, & {k=0,} \\
    1, & {k=1.}
  \end{array}
\right.$$
Therefore, $H_d(k)/ \gcd(d,N)$ is a PP over $\mathbb{Z}_{2}$.

{\bf 2)}  Then we prove that $H_d'(k)/ \gcd(d,N) \not\equiv 0  \,\,(\mathrm{mod}\,\,2)$ for any $k \in \mathbb{Z}_{2}$.
Note that
$H_d'(k)/ \gcd(d,N)=  \frac{ \pi(k+d) \pi'(k+d) -\pi(k) \pi'(k) }{\gcd(d,N)}$.
When $k=1$,
since $\frac{d^2 }{\gcd(d,N)} \,\,(\mathrm{mod}\,\,2)
\equiv 0$, we have
\begin{align}
H_d'(1)/ \gcd(d,N)
= &
  \frac{1}{\gcd(d,N)} \left[
\sum\limits_{i=1}^{m}  a_i  \sum\limits_{l=1}^{m} l a_l(1+d)^{i+l-1}
-\sum\limits_{i=1}^{m}  a_i\sum\limits_{l=1}^{m} l a_l 1^{i+l-1} \right] \notag  \\
\equiv    &
 \sum_{i=1}^{m} \sum_{l=1}^{m}  a_i a_l l \binom{i+m-1}{1} 1^{i+m-2}
 \frac{d}{\gcd(d,N)}     \,\,(\mathrm{mod}\,\,2)   \notag  \\
\equiv  &
\sum_{i=1}^{m} \sum_{l=1}^{m}  a_i a_l  l(i+l-1) 1^{i+m-2} \,\,(\mathrm{mod}\,\,2)  \notag \\
%\end{align}
%Let $k=1$, then
%\begin{align}
%H_d'(k)/ \gcd(d,N)
%= &
%\sum_{i=1}^{m} \sum_{l=1}^{m}  a_i a_l  l(i+l-1)  \,\,(\mathrm{mod}\,\,2)  \notag  \\
%= &\sum_{ i=1}^{m} \sum_{\text{ odd } l=1}^{m}  a_i a_l  l(i+l-1)  \,\,(\mathrm{mod}\,\,2)
=&
\sum_{\text{ odd } i=1}^{m}  a_i  \sum_{\text{ odd } l=1}^{m}  a_l    \,\,(\mathrm{mod}\,\,2)
=  1   \,\,(\mathrm{mod}\,\,2) \notag
\end{align}
where the last step holds by Lemma \ref{N=2^n}.
When $k=0$, due to $\pi(0)=0$, we have
\begin{align}
& H_d'(0)/ \gcd(d,N)=  \frac{\pi(d) \pi'(d) }{\gcd(d,N)}   \notag \\
= & \frac{d }{\gcd(d,N)}   (a_1 + a_2d  + a_{3}d^{2} + \dots +  a_m d^{m-1}  )
 (a_1 + 2a_2d  + 3a_{3}d^{2} + \dots + m a_m d^{m-1}  )    \notag \\
\equiv & a_1^{2}  \,\,(\mathrm{mod}\,\,2) \equiv 1  \,\,(\mathrm{mod}\,\,2). \notag
\end{align}
Therefore, $H_d'(k)/ \gcd(d,N) \not\equiv 0  \,\,(\mathrm{mod}\,\,2)$ for any $k \in \mathbb{Z}_{2}$.

{\bf Case 2.}
 $\gcd(d,2)=1$. Then $H_d(k)/\gcd(d,N)=H_d(k)$.

{\bf 1)} We first prove that $H_d(k)$ is a PP over $\mathbb{Z}_{2}$.
%Since $\pi$ is a PP over $\mathbb{Z}_{2^n}$,
%from Lemma \ref{N=p^n} (i), we know that
%$\pi$ is a PP over $\mathbb{Z}_{2}$.
%Thus for $\pi(k)=a_lk^l + a_{l-1}k^{l-1} + \dots + a_2k^2 + a_1k$.
%it follows from $\pi(0)=0$ that  $\pi(1)\equiv 1  \,\,(\mathrm{mod}\,\,2)$.
It follows from Lemma \ref{N=2^n} that
$\pi(1)= \sum\limits_{i=1}^{m}a_i  \equiv 1  \,\,(\mathrm{mod}\,\,2)$.
Hence,
$\pi(d)=\pi(2d_1+1) \equiv 1  \,\,(\mathrm{mod}\,\,2)$
and $\pi(1+d)=\pi(2(d_1+1)) \equiv 0  \,\,(\mathrm{mod}\,\,2)$.
Therefore, we can set
 \begin{equation}\label{pi_value}
\pi(1)=2l_1 +1, \pi(d)=2l_2 +1, \text{ and } \pi(1+d)=2l_3
 \end{equation}
with integers $l_i, i=1,2,3$.
Thus
$$H_d(1) =\frac{\pi^2(1+d) -\pi^2(1)-\pi^2(d) }{2}
= \frac{1}{2} [ (2l_3)^2 -(2l_1 +1)^2 - (2l_2 +1)^2     ]
\equiv 1  \,\,(\mathrm{mod}\,\,2).$$
Together with $H_d(0)=0$, it implies  that  $H_d(k)$ is a PP over $\mathbb{Z}_{2}$.

{\bf 2)} Now we prove that $H_d'(k) \not\equiv 0  \,\,(\mathrm{mod}\,\,2)$ for any $k \in \mathbb{Z}_{2}$.
We can derive that
$H_d'(k) =  \pi(k+d) \pi'(k+d)-\pi(k) \pi'(k) $.
Since $\pi$ is a PP over $\mathbb{Z}_{2^n}$,
from Lemma \ref{N=p^n} (i), we know that
$\pi'(k) (\mathrm{mod}\,\,2) \not\equiv 0$ for $k\in \mathbb{Z}_{N}$.
Combining with \eqref{pi_value}, it follows that
$$H_d'(0) = \pi(d) \pi'(d) -\pi(0) \pi'(0)  \equiv    \pi'(d)    \,\,(\mathrm{mod}\,\,2) \not\equiv 0  \,\,(\mathrm{mod}\,\,2),  $$
and
$$H_d'(1) = \pi(1+d) \pi'(1+d)- \pi(1) \pi'(1)   \equiv  - \pi'(1)    \,\,(\mathrm{mod}\,\,2) \not\equiv 0  \,\,(\mathrm{mod}\,\,2).  $$
That is to say,  $H_d'(k)    \not\equiv 0 \,\,(\mathrm{mod}\,\,2)$.

Together with Cases 1 and 2, we can see that  $H_d(k)/ \gcd(d,N)$ is a PP over $\mathbb{Z}_{N/ \gcd(d,N)}$.
It follows that $ \sum^{N-1}\limits_{k=0} \zeta_{N/ \gcd(d,N)}^{ H_d(k) / \gcd(d,N)  }=0$,
implying  $\theta(d)=0$ in \eqref{corred}.  The desired conclusion thus
follows.
\end{proof}

\section*{Appendix D: Lemma \ref{Apend_thm2_PP} }\label{appendix:D}
\begin{lem}\label{Apend_thm2_PP}
Let $N=p^nN_1$ with $n\geq 2$,  $N_1= q_1q_2\dots q_r$, different primes $p\geq 3$ and $q_i$ satisfying $p-1 \equiv 0 \, (\mathrm{mod}\,\, q_i -1)$ for all $1 \leq i \leq r$.
And let
 $\gcd(u,N)=1$, $d \in \mathbb{Z}_N$, $m \in \mathbb{Z}_N$ and
$$h_d(k)= (  (k+d)^p-k^p-d^p 	)m +udk . $$
Then
$h_d(k)/ \gcd(d,N)$ is a PP over $\mathbb{Z}_{N/ \gcd(d,N)}$.
\end{lem}
\begin{proof}
According to Lemma \ref{N=p^n} (ii), $H_d(k)/ \gcd(d,N)$ is a PP over $\mathbb{Z}_{N}$ if and only if  $H_d(k)/ \gcd(d,N)$ is a PP over $\mathbb{Z}_{\frac{p^n}{\gcd(d,p^n)}}$ and $H_d(k)/ \gcd(d,N)$ is a PP over $\mathbb{Z}_{\frac{q_i}{\gcd(d,q_i)}}, i=1,2,\dots,r$.

{\bf Case 1.} We prove that $H_d(k)/ \gcd(d,N)$ is a PP over $\mathbb{Z}_{\frac{p^n}{\gcd(d,p^n)}}$.
From Lemma \ref{N=p^n} (i),
it suffices to investigate that $H_d(k)/ \gcd(d,N)$ is a PP over $\mathbb{Z}_{p}$ and its derivative $H_d'(k)/ \gcd(d,N) \not\equiv 0  \,\,(\mathrm{mod}\,\,p)$ for any $k \in \mathbb{Z}_{\frac{p^n}{\gcd(d,p^n)}}$.

We first prove that $H_d(k)/ \gcd(d,N)$ is a PP over $\mathbb{Z}_{p}$.
In fact, we have
{\small
\begin{align*}%\label{}
\frac{ H_d(k) }{\gcd(d,N)} =
\frac{d }{\gcd(d,N)} \left(  m \sum\limits_{i=1}^{p-1} \binom{p}{i}k^{p-i} d^{i-1} +uk    \right)      \equiv \frac{d }{\gcd(d,N)} uk     \,\,(\mathrm{mod}\,\,p),
\end{align*}
}where the last step follows from $p \mid \binom{p}{i}, 1\leq i<p$.
Thus we can see that
$\frac{ H_d(k) }{\gcd(d,N)}$
 is a PP over $\mathbb{Z}_{p}$ since  $\frac{d }{\gcd(d,N)}$ and $u$ are coprime to $p$.
Then we prove that $H_d'(k)/ \gcd(d,N) \not\equiv 0  \,\,(\mathrm{mod}\,\,p)$ for any $k \in \mathbb{Z}_{\frac{p^n}{\gcd(d,p^n)}}$.
% Similarly as the analysis in \eqref{p_zhengchu_d},
We can deduce that
{\small
$$H_d'(k)/ \gcd(d,N) = mp (\frac{  (k+d)^{p-1} -k^{p-1}  }{\gcd(d,N)} ) + u\frac{d }{\gcd(d,N)}
\equiv  u\frac{d }{\gcd(d,N)} \,\,(\mathrm{mod}\,\,p), $$
}
which is nonzero from $\gcd(u\frac{d }{\gcd(d,N)},p)=1$.

{\bf Case 2.} We prove that $H_d(k)/ \gcd(d,N)$ is a PP over $\mathbb{Z}_{\frac{q_i}{\gcd(d,q_i)}}, i=1,2,\dots,r$.
For each $i$,
it is obvious for the case of $\gcd(d,q_i)=q_i$,
then
it suffices to consider the case of  $\gcd(d,q_i)=1$.
Hence we have $\gcd( \gcd(d,N),q_i)=1$.
Thus for each $i$, there exist a positive integer $t_i$ with  $\gcd(t_i,q_i)=1$ such that
$\frac{1 }{\gcd(d,N)} \equiv  t_i \,\,(\mathrm{mod}\,\,q_i)$.

Since $p-1 \equiv 0 \, (\mathrm{mod}\,\, q_i -1)$ for all $1 \leq i \leq r$,
we have $x^p \equiv x \, (\mathrm{mod}\,\, q_i)$ for $x\in \mathbb{Z}_{q_i}$.
Then it follows that
$$
%  \frac{H_d(k)}{\gcd(d,N)}
H_d(k)/ \gcd(d,N)
 \equiv   t_iH_d(k)
= t_i [((k+d)^p-d^p-k^p  )m+ udk]
 \equiv   t_iudk   \,\,(\mathrm{mod}\,\,q_i).
$$
Thus  $\frac{H_d(k)}{\gcd(d,N)}$   is a PP over $\mathbb{Z}_{q_i}$, as $t_i$,
$u $ and $d$ are all coprime to $q_i$.
Moreover, for any $k$,  we can deduce that
$$H_d'(k)/ \gcd(d,N)\equiv   t_iH_d'(k) = mpt_i [(d+k)^{p-1} - k^{p-1} ] + t_iud
\equiv  t_iud \,\,(\mathrm{mod}\,\,q_i), $$
which is nonzero from $\gcd(t_iud,p)=1$.

Combining the above two cases, we get $H_d(k)/ \gcd(d,N)$ is a PP over $\mathbb{Z}_{N/ \gcd(d,N)}$. The desired conclusion thus follows.
\end{proof}

\section*{Appendix E:  Proof of Theorem \ref{mainresult2} (ii)}\label{appendix:E}
\begin{proof}
According to Lemma \ref{CA_ZAC}, it suffices to show that
$ \widehat{y}(t)\widehat{y}^*(t)$ in \eqref{xx*first} is a constant for any $0\leq t<N$.
By substituting the PP $\pi(x)=a_mx^m + a_{m-1}x^{m-1} + \dots + a_2x^2 + a_1x + a_0$
 into $M(d)$ in \eqref{MMMvv}, it  yields
$$M(d)=\sum\limits_{k=0}^{N-1}\xi_N^{- udk-(\pi (k+d)-\pi(k)-\pi(d) +a_0	)m }
=\sum\limits_{k=0}^{N-1}\xi_N^{-h_d(k)}=\sum_{k=0}^{N-1}\xi_{N/ \gcd(d,N)}^{- h_d(k)/ \gcd(d,N)} ,$$
where $\gcd(u,N)=1$, $d \in \mathbb{Z}_N$, $m \in \mathbb{Z}_N$ and
$$h_d(k)=   \sum\limits_{i=2}^{m} a_i( (k+d)^i-k^i-d^i 	)m +udk . $$
We now discuss $M(d)$ in two cases according to the values of $d$.
When $d=0$, it is evident that $M(d)=N$.
When $d\neq 0$, we shall prove that $h_d(k)/ \gcd(d,N)$ is a PP over $\mathbb{Z}_{N/ \gcd(d,N)}$.

According to Lemma \ref{N=p^n} (i), $\frac{H_d(k)}{\gcd(d,N)}$ is a PP over $\mathbb{Z}_{2^n}$ if and only if  $\frac{H_d(k)}{\gcd(d,N)}$ is a PP over $\mathbb{Z}_{2}$ and its derivative $H_d'(k)/ \gcd(d,N) \not\equiv 0  \,\,(\mathrm{mod}\,\,2)$ for any $k \in \mathbb{Z}_{2}$.
In the following, the discussion is divided into two cases $2\,|\, d$
and $\gcd(2,d)=1$.

{\bf Case 1.}   $2\,|\, d$.
{\bf 1)} We first prove that $H_d(k)/ \gcd(d,N)$ is a PP over $\mathbb{Z}_{2}$.
Since $\frac{d^i}{\gcd(d,N)} \equiv 0  \,\,(\mathrm{mod}\,\,2) $ with $i \geq 2$,
$\frac{d}{\gcd(d,N)} \equiv 1  \,\,(\mathrm{mod}\,\,2) $ and $\gcd(u,2)=1$, it follows
\begin{align*}%\label{}
  H_d(k)/ \gcd(d,N)  &=
u\frac{d}{\gcd(d,N)}k +
\sum\limits_{i=2}^{m} a_i  \frac{ (k+d)^i-k^i-d^i}{\gcd(d,N)}  m
 \notag \\
  &  \equiv u\frac{d}{\gcd(d,N)}k +
\sum\limits_{i=2}^{m} a_i  \frac{ \binom{i}{1}dk^{i-1}}{\gcd(d,N)}  m  \,\,(\mathrm{mod}\,\,2)
 \notag \\
  &  \equiv k +
\sum\limits_{i=2}^{m} a_i m  i k^{i-1}  m   \,\,(\mathrm{mod}\,\,2).
\end{align*}
Moreover, since
\begin{equation}\label{ami}
  \sum\limits_{i=2}^{m} a_i   m i
\equiv  \sum\limits_{\text{ odd } i=2}^{m} a_i m  \,\,(\mathrm{mod}\,\,2)
\equiv \left\{
  \begin{array}{ll}
     0\,\,(\mathrm{mod}\,\,2), & \text{ even } m, \\
    \sum\limits_{\text{ odd } i=2}^{m} a_i \equiv 0\,\,(\mathrm{mod}\,\,2) , & \text{ odd } m,
  \end{array}
\right.
\end{equation}
where the last step holds by Lemma \ref{N=2^n},
we can derive $\frac{H_d(0)}{\gcd(d,N)}=0$ and
$\frac{H_d(1)}{\gcd(d,N)} \equiv  1 \,\,(\mathrm{mod}\,\,2)$.
It implies that $H_d(k)/ \gcd(d,N)$ is a PP over $\mathbb{Z}_{2}$.

{\bf 2)}  Then we prove that $H_d'(k)/ \gcd(d,N) \not\equiv 0  \,\,(\mathrm{mod}\,\,2)$ for any $k \in \mathbb{Z}_{2}$.
Note that
$H_d'(k)/ \gcd(d,N)=  u\frac{d}{\gcd(d,N)} +
\sum\limits_{i=2}^{m} a_i \,  i\frac{ (k+d)^{i-1}-k^{i-1}}{\gcd(d,N)}  m $.
Similarly as {\bf 1)}, it derives that
$H_d'(k)/ \gcd(d,N)\equiv 1  \,\,(\mathrm{mod}\,\,2) $ for $k \in \mathbb{Z}_{2}$.

{\bf Case 2.}
 $\gcd(d,2)=1$. Then $H_d(k)/\gcd(d,N)=H_d(k)$.

{\bf 1)} We first prove that $H_d(k)$ is a PP over $\mathbb{Z}_{2}$.
Since $1+d$ is even and $d$ is odd, we have
$$H_d(1) = ud + \sum\limits_{i=2}^{m} a_i  \frac{ (1+d)^i-1-d^i}{\gcd(d,N)}  m
\equiv 1  \,\,(\mathrm{mod}\,\,2).$$
Together with $H_d(0)=0$, it implies that  $H_d(k)$ is a PP over $\mathbb{Z}_{2}$.

{\bf 2)} Now we prove that $H_d'(k) \not\equiv 0  \,\,(\mathrm{mod}\,\,2)$ for any $k \in \mathbb{Z}_{2}$.
We can derive that
$H_d'(k) =  ud + \sum\limits_{i=2}^{m} a_i \, i m[ (k+d)^{i-1}-k^{i-1}]  $.
From \eqref{ami}, it follows that
$$H_d'(0) = ud +\sum\limits_{i=2}^{m} a_i \, i m  d^{i-1} \equiv  1   \,\,(\mathrm{mod}\,\,2).$$
Since $1+d$ is even and \eqref{ami}, we get
$$H_d'(1) = ud + \sum\limits_{i=2}^{m} a_i \, i m[ (1+d)^{i-1}-1]
\equiv ud - \sum\limits_{i=2}^{m} a_i \, i m \,\,(\mathrm{mod}\,\,2) \equiv  1   \,\,(\mathrm{mod}\,\,2).  $$
That is $H_d'(k)    \not\equiv 0 \,\,(\mathrm{mod}\,\,2)$ for any $k \in \mathbb{Z}_{2}$.

Together with Cases 1 and 2, when $d\neq 0$, we can see that  $H_d(k)/ \gcd(d,N)$ is a PP over $\mathbb{Z}_{N/ \gcd(d,N)}$.
It implies that $ M(d)=\sum^{N-1}\limits_{k=0} \xi_{N/ \gcd(d,N)}^{ H_d(k) / \gcd(d,N)  }=0$
when $d\neq 0$.
According to \eqref{xx*first}  and $ M(0)=N$, it follows that
\begin{align*}
	 \widehat{s}(m)\widehat{s}^*(m)
	&{=}\frac{\xi_N^{a_0m}}{N} \, \sum_{d=0}^{N-1}\xi_N^{-\pi^{-1} (d)m
-ud( d+ (N  \, \mathrm{mod}\,\,2)+2l)/2}
	 \, M(d) \notag\\
	&{=}\xi_N^{a_0 m}\xi_N^{- \pi^{-1} (0)m}  \\
	&{=}1, \notag
\end{align*}
where the last step holds by $\pi^{-1} (0)=a_0 $.
Hence, the proof is finished.
\end{proof}

\section*{Appendix F:  Proof of Proposition \ref{pro2}}\label{appendix:F}
%\section*{Appendix F:  Proof of Proposition \ref{prop1}}% \label{Appendix-prop2}	
\begin{proof}
{\bf Case 1.} Assume equation \eqref{ex_F_case1} holds.
%For the case of $N \equiv 1 \,\,(\mathrm{mod}\,2) $,
The proofs for odd and even
$N$ are analogous.
%The proofs for $N$ being odd and for $N$ being even follow similarly.
Thus, we provide only the detailed proof for the case of odd $N$.
Since $\gcd(2,p^n)=1$,
we obtain that
\begin{equation}\label{tongyu}
	F(k) \equiv 0
\,\,(\mathrm{mod}\,\,p^n) \text{ for any }  k\in \mathbb{Z}_{p^n} \text{ with } n\geq 2.
\end{equation}

%Consider \eqref{tongyu} with $i=0$ holds for any $k\in \mathbb{Z}_{p^n}$ with $n\geq 2$.
On the one hand, it follows
$	F(k) \equiv 0 	\,\,(\mathrm{mod}\,\,p).$
Let $k=k_0+k_1p$ with $k_0\in \mathbb{Z}_{p}$ and $k_1\in \mathbb{Z}_{p^{n-1}}$.
Then $k^p \equiv k_0^p \equiv k_0 \equiv k	\,\,(\mathrm{mod}\,\,p)$.
Moreover since $f_2 \equiv 0 	\,\,(\mathrm{mod}\,\,p)$, \eqref{tongyu} is reduced to
\begin{equation}\label{BBBNNN222}
	F(k) \equiv b_2k^2 + b_1k + b_0  \equiv 0 \,\,(\mathrm{mod}\,\,p) \text{ for any } k\in \mathbb{Z}_{p},
\end{equation}
where $	b_2 = u_1 (a+1)^2 -su_2 f_1^2$,
$	b_1 = u_1 (a+1)(2b+1) -2v -su_2(2f_1d +f_1)$  and $	b_0 =u_1 b(b+1)-su_2(d^2+d ) -r$.
Similar as Proposition \ref{prop1}, equation  \eqref{BBBNNN222} with $k=0,1,2$ can be rewritten as $B'(b_2,b_1,b_0)^T \equiv 0 \,\,(\mathrm{mod}\,\,p) $,
where the matrix $B'$ is a Vandermonde matrix with $|B'|=\prod\limits_{0\leq i<j \leq 2}(j-i)$.
It implies that $\gcd(|B'|,p)=1$ by prime $p\geq 3$.
Based on Lemma \ref{solution_lem} and $\gcd(|B'|,p)=1$, \eqref{BBBNNN222} has a unique set of incongruent solutions
$b_i  \equiv 0  \,\,(\mathrm{mod}\,\, p), 0\leq i\leq 2$.
Thus
\begin{equation}\label{FFF1}
	b_2 = u_1 (a+1)^2 -su_2 f_1^2 \equiv 0  \,(\mathrm{mod}\,\, p) .
\end{equation}

On the other hand, from \eqref{tongyu}   we have
$	F(k) \equiv 0 \,\,(\mathrm{mod}\,\,p^2).$
 Since $f_2 \equiv 0 	\,\,(\mathrm{mod}\,\,p)$,
it follows that
\begin{equation*}%\label{BBBNNN}
	\begin{array}{cll}
	F(k) \equiv & u_1 ( k^{2p} + 2a k^{p+1} +  k^{p} + a^2 k^{2} + a k )+u_1 ( 2bk^{p} + 2ab k +  b^2 + b ) -2vk  \\
&	  - su_2 [  2f_2f_1 k^3 +(f_1^2+2f_2d +f_2 )k^2 +(2f_1d +f_1)k +d^2+d      ] 	
  -r \equiv 0 \,\,(\mathrm{mod}\,\,p^2).
\end{array}
\end{equation*}
Since $F(0) \equiv 0 \,\,(\mathrm{mod}\,\,p^2)$ and  $F(p) \equiv 0 \,\,(\mathrm{mod}\,\,p^2)$,
it follows that
$$	F(p)-F(0) = u_1ap(2b+1) -2vp-su_2(2f_1d+f_1)p  \equiv 0 \,\,(\mathrm{mod}\,\,p^2),$$
  implying
\begin{equation}\label{FF00}
u_1a(2b+1)-2v-su_2f_1(2d+1) \equiv 0 \,\,(\mathrm{mod}\,\,p).
\end{equation}
 Moreover, since $(1+p)^{2p} \equiv 1 \,\,(\mathrm{mod}\,\,p^2)$ and $(1+p)^{p+1} \equiv 1+p \,\,(\mathrm{mod}\,\,p^2)$,
 we can get that
$$
\begin{array}{cll}
&F(1+p)-F(1)\\
\equiv \!\!& u_1ap(2a +2b+3)-2vp-su_2[6f_2f_1p + 2p(f_1^2 + 2f_2d +f_2) + (2f_1d+f_1)p ] \\
\equiv  \,\,& 0\,\, (\mathrm{mod}\,\,p^2).
\end{array}
$$
Due to $p\,|\, f_2$, it implies that
$ u_1a(2a+2b+3)-2v - su_2f_1(2f_1 +2d+1) \equiv 0 \,\,(\mathrm{mod}\,\,p).$
Together with \eqref{FF00}, since $\gcd(2,p)=1$, we  see that
\begin{equation}\label{FFF2}
 u_1a(a+1) - su_2f_1^2 \equiv 0 \,\,(\mathrm{mod}\,\,p).
\end{equation}

Consequently, combining \eqref{FFF1}  and  \eqref{FFF2}, it follows $p \,|\, u_1a(a+1)$,
which contradicts that $(u_1,p) =1$ and $a \not\equiv 0,-1\,\,(\mathrm{mod}\,\,p). $
Therefore, the assumption does not hold.
%For the case of $N \equiv 0 \,\,(\mathrm{mod}\,2) $, it can be proved similarly.

{\bf Case 2.} Assume equation \eqref{ex_F_case2} holds.
Then similarly as Case 1, we can prove that there is a contradiction.

Therefore, $\s \circ \pi$ are inequivalent to
 QPPs and their inverses interleaved ZC sequences.
\end{proof}

\section*{Appendix G:  Proof of Lemma \ref{one_parsum}}\label{appendix:G}
\begin{proof}
Let $\beta=\frac{b}{q}$ with $\gcd(b, q)=1$, and let $k'=\frac{k}{g}$. Then  $\gcd(k'a, q)=1$.
It implies that
$$
 \sum\limits_{x=0}^{Q-1}  \min(P, \frac{1}{2\| k \alpha x + \beta \|})
 = \sum\limits_{x=0}^{Q-1}  \min(P, \frac{1}{2\| \frac{k' a  }{q} g x +\frac{b}{q} \|})
 < \sum\limits_{x=0}^{g Q-1}  \min(P, \frac{1}{2\| \frac{k'a x +b}{q}  \|}).
$$
Let   $Q'= 1+ \lfloor \frac{gQ}{q}\rfloor$.
 By replacing $x$ by $qx_1+x_2$, it follows that
\begin{align*}%\label{fourth_esm}
\sum\limits_{x=0}^{ gQ-1}  \min(P, \frac{1}{2\| \frac{k'a x +b}{q}  \|})
& \leq  \sum\limits_{x_1=0}^{Q'-1}  \sum\limits_{x_2=0}^{q-1}  \min(P, \frac{1}{2\| \frac{k'a x_2 +b}{q} + k'a x_1  \|})  \nonumber \\
% & \leq Q' \sum\limits_{x_2=1}^{q}  \min(P, \frac{1}{2\| \frac{k'a x_2 +\beta}{q} \|})   \nonumber \\
&    \leq( 1+  \frac{gQ}{q}) \sum\limits_{x_2=0}^{q-1}  \min(P, \frac{1}{2\| \frac{k'a x_2 +b}{q} \|}).
\end{align*}
Due to $\gcd(k'a, q)=1$, $k'a x_2 +b$ runs through a complete residue system modulo $q$,
when $x_2$ runs through a complete residue system modulo $q$. Thus,
\begin{align*}%\label{fifth_esm}
\sum\limits_{x_2=0}^{q-1}  \min(P, \frac{1}{2\| \frac{k'a x_2 +b}{q} \|})
&= \sum\limits_{x=0}^{q-1}  \min(P, \frac{1}{2\| \frac{x}{q}  \|})     \nonumber \\
&= P+ \sum\limits_{x=1}^{q-1}  \min(P, \frac{1}{2\| \frac{x}{q} \|})    \nonumber \\
&= P+ 2 \sum_{1 \leq x<q/2}  \min(P, \frac{q}{2x })  +1  \nonumber \\
&= P+1 + 2 \sum_{1 \leq x \leq \frac{q}{2P}} P +  q \sum_{\frac{q}{2P}+1 \leq x<\frac{q}{2} }  \frac{1}{x}      \nonumber \\
&\leq  P+1 + 2 \frac{q}{2P}P +  q  \int_{\frac{q}{2P}}^{ \frac{q}{2}}  \frac{1}{x} dx     \nonumber \\
&\leq  P+1 + q +  q  \ln P    \nonumber \\
& <2(P+q)\frac{P^\varepsilon}{\varepsilon},
\end{align*}
where the last step follows from $1 < \ln P <  \frac{P^\varepsilon}{\varepsilon}$ by $P \geq 3$.

Therefore, we derive that
%According to the condition $P^{\varepsilon_1} \leq q \leq P^{3-\varepsilon_1} (0< \varepsilon_1 <1) $,
%combining \eqref{thrid_esm}, \eqref{fourth_esm} and \eqref{fifth_esm} yields
\begin{equation*}
\begin{array}{c}
\sum\limits_{x=0}^{Q-1}  \min(P, \frac{1}{2\| k \alpha x + \beta \|})
 \leq
  ( 1+  \frac{gQ}{q}) (P+q)\frac{2P^\varepsilon}{\varepsilon}.
 %\leq \frac{8C_2(\varepsilon)}{\varepsilon 4^{\varepsilon}}  P^{3-\varepsilon_1+3\varepsilon},
 \end{array}
\end{equation*}
\end{proof}
\nn

%\end{CJK*}	
	
\end{document}